\documentclass[letterpaper,11pt]{article}
\usepackage[utf8]{inputenc}
\usepackage{amsmath}
\usepackage{amsthm, amsfonts, amssymb}
\usepackage{breqn}
\usepackage{graphicx}
\usepackage{mathtools}
\usepackage{appendix}
\usepackage[margin = 1in]{geometry}
\usepackage[section]{placeins}
\usepackage{float}
\usepackage{algorithm, algorithmicx}
\usepackage{comment}
\usepackage{hyperref}
\usepackage{xcolor}
\usepackage{enumitem}
\usepackage{yhmath}
\usepackage[skip = 1.25pt]{caption, subcaption}

\usepackage[T1]{fontenc}   %  font for Leszek's name
\usepackage[square,sort,comma,numbers]{natbib}
\newtheorem{defin}{Definition}
\newtheorem{theorem}{Theorem}
\newtheorem{lemma}[theorem]{Lemma}

\newtheorem{prop}[theorem]{Proposition}
\newtheorem{remark}[theorem]{Remark}

\title{Fast Two-Robot Disk Evacuation with Wireless Communication}
\author{
  Ioannis Lamprou, %\thanks{Corresponding Author, E-mail: \texttt{Ioannis.Lamprou@liverpool.ac.uk}},
  Russell Martin,
  Sven Schewe
  \\
  \texttt{\{Forename.Surname\}@liverpool.ac.uk}
  \\ \\
  Department of Computer Science,
  \\
  University of Liverpool, UK
  }
  
\begin{document}

\maketitle
\thispagestyle{empty}

\begin{abstract}
In the fast evacuation problem, we study the path planning problem for two robots who want to minimize the worst-case evacuation time on the unit disk.
The robots are initially placed at the center of the disk.
In order to evacuate, they need to reach an unknown point, the exit, on the boundary of the disk.
Once one of the robots finds the exit, it will instantaneously notify the other agent, who will make a beeline to it. 

The problem has been studied for robots with the same speed~\cite{s1}.  We study a more general case where one robot 
has speed $1$ and the other has speed $s \geq 1$.  We provide optimal evacuation strategies in the case that
$s \geq c_{2.75} \approx 2.75$ by showing matching upper and lower bounds on the worst-case evacuation time.  For 
$1\leq s < c_{2.75}$, we show (non-matching) upper and lower bounds on the evacuation time with a ratio less than $1.22$.  
Moreover, we demonstrate that a generalization 
of the two-robot search strategy from~\cite{s1} is outperformed by our proposed strategies for any $s \geq c_{1.71} \approx 1.71$.  

% The problem has been studied for robots with the same speed.
% We study the general case, where the robots might have different speed.
% We provide optimal bounds for the case where only one robot explores the boundary of the disk in search of the exit.
% This search strategy proves to be superior to the generalization of the optimal search strategy for two agents with similar speed when the factor between the different speeds exceeds $1.86$.
% We establish that this search technique is optimality when this factor exceeds $4.97$.
% For the smaller speed factors, we show that the generalization of the strategy for robots using this strategy is below $1.xyz$.

% The robots are capable of transmitting information to one another at any time via wireless communication.
% The task in question is to minimize the worst-case \emph{evacuation time}, i.e.\ the time until both robots find the exit.
% To do so, the robots need to adapt clever \emph{evacuation strategies}, which define their movements until the exit is discovered by either of them.
% A matching lower and upper bound is already known for the case $s = 1$ \cite{s1}.
% We consider a partition of the strategy space based on which robot(s) explore the disk boundary and provide lower and upper bounds for each subset of strategies.
% Overall, our evacuation strategies are proved to be optimal for $s \ge c$, where $c \approx 4.9699$, while we provide lower and upper bounds for $s \in (1, c)$. 

\end{abstract}

\newpage

\section{Introduction}
\setcounter{page}{1}
Consider a pair of mobile robots in an environment represented by a circular
disk of unit radius.  The goal of the robots is to find an {\em exit}, i.e.\ a point at an
unknown location on the boundary of the disk, and both move to this exit.  The exit
is only recognized when a robot visits it.  The robots' aim is to accomplish this task 
as quickly as possible. This problem is referred to as the 
{\em evacuation problem}.   The robots start at the center of the disk and can move
with a speed not exceeding their maximum velocity (which may be different from
one another).  They can coordinate their actions
in any manner they like, and can communicate wirelessly (instantaneously).    

\subsection{Related work}

Evacuation belongs to the realm of distributed search problems, which have 
a long history in mathematics, computer science, and operations 
research, see, e.g.~\cite{B64,B56,B63}.  

Salient features in search problems include the {\em environment} (e.g.\ a
geometric one or graph-based), {\em mobility} of the robots (how they are allowed
to move), {\em perception} of and {\em interaction} with the environment, and 
their {\em computational}
and {\em communication abilities}.  Typical tasks include exploring and mapping
an unknown environment, finding a (mobile or immobile) target (e.g.\ cops and
robbers games~\cite{Bonato} and pursuit-evasion games~\cite{Parsons}; 
the ``lost at sea'' problem~\cite{Gluss}; 
the cow-path problem and plane-searching 
problem~\cite{BCR88,BCR93,BDD, BS95, JL09, KRT, LC09, TF10}), 
rendezvous or gathering
of mobile agents~\cite{KDM,KKR}, and evacuation~\cite{CGGM,s1,face-to-face}.  
(Note that we distinguish between the distributed version of evacuation problems
involving a search for an unknown exit, and centralized versions, typically modeled 
as (dynamic) capacitated flow problems on graphs, where the exit is known.)  
A general survey of search and rendezvous problems can be found in~\cite{alpern02b}.
Also related is the task of patrolling or monitoring, i.e.\ the
periodic (re)visitation of (part of) the environment~\cite{YC,czyz2011,YWB}.  

In most all of these settings, the typical cost is the time required to finish the task
(in a synchronous environment), or the total distance moved by the robots to 
finish it (in an asynchronous setting).  (Patrolling has a different ``cost'', that being
the time between consecutive visits to any point in the region, the so-called 
``idle time''.)

% While the main task of mobile agents may be different from exploration, it is
% often necessary for them to devote some of their activity to recognizing or mapping
% their environment, so search problems are a natural ``pre-processing'' step 
% in their activities.  Much literature has been devoted to studying distributed mechanisms
% for mobile agents, to determine feasibility of search, memory or computational
% power required for the task, consideration of different communication abilities (e.g.\
% wireless, ``whiteboards'' to leave messages), and the use of 
% external markers or ``pebbles'' to aid in the distributed process~\cite{Bender}.  

A little explored feature of the robots is their {\em speed}.  Most past work has
focused on the case where the robots all share the same (maximal) speed.  Notable
exceptions of which the authors are aware include~\cite{CGGM} which considers 
the evacuation problem on the infinite line with robots with distinct maximal speeds, 
and \cite{czyz2011} where the
authors show a non-intuitive ring patrolling strategy using three
robots with distinct maximal speeds.  It is this feature, robots with 
different maximal speeds, that we explore in this paper.  

The most relevant previous work is~\cite{s1,face-to-face}, which explores the evacuation
problem in the unit disk with two robots with identical speeds ($s=1$).  

\subsection{Our results}

We consider the evacuation problem in the unit disk using two robots with distinct maximal speeds
(one with speed $1$, the second with speed $s \geq 1$).
% and where the robots 
% use wireless communication.  The robots share a common clock, so can 
The robots share a common clock and can communicate instantaneously when they have found the exit
(wireless communication) and so can synchronize their behavior
in the evacuation procedure.
We assume that the robots can measure distances to an arbitrary precision
(equivalently, they can measure time to an arbitrary precision), and can vary
their speeds as they desire, up to their maximum speed. 

We show that even in the case of two robots, the analysis involved in finding
(time) optimal evacuation strategies can become intricate, 
with strategies that depend on the
(ratio of) the fast robot's maximal speed.
For large $s$, we introduce an efficient search strategy, called the 
Half-Chord Strategy (Figure~\ref{fig:chord}).
We 
generalize a strategy from~\cite{s1} for small $s$, the ``Both-to-the-Same-Point Strategy'' (BSP), where
the two robots move to the same point on the boundary and then separately explore the boundary
in clockwise and counterclockwise directions to find the exit (Figure~\ref{fig:bsp}). For values of
$s \geq c_{1.86}$ (with $c_{1.86} \approx 1.856$), we show BSP is not optimal by 
demonstrating that the Half-Chord Strategy is superior to it.  
Moreover, we improve on this with the Fast-Chord Strategy (Figure~\ref{fig:fast-chord}), which outperforms Half-Chord for $1.71 \approx c_{1.71} < s < c_{2.07} \approx 2.07$.
We obtain optimality
for all $s \geq c_{2.75} \approx 2.75$, in the wireless setting, 
as we demonstrate matching upper
and lower bounds on the evacuation time. 
On the other hand, for $s \in (1, c_{2.75})$ we provide
lower bounds on the evacuation time that do not match the bounds provided by the 
respective search strategies (BSP for $s < c_{1.71}$, Fast-Chord for $s \in [c_{1.71}, c_{2.07})$ and Half-Chord for $s \geq c_{2.07}$).

Section~\ref{sec:definition} contains a more formal definition of the problem
we consider.  Section~\ref{sec:upper} contains our upper bounds on the
evacuation time, while Section~\ref{sec:lower} has our lower bounds.  
In the interests of space, parts of the proofs are omitted from this version, and
we trust the reader to rely upon the supplied diagrams for the intuition of
our results.  

\section{Problem Definition and Strategy Space}\label{sec:definition}

In this section, we formally define the problem in question.
Furthermore, we provide a partition of the strategy space and some observations, which will be useful in the bounds to follow.

\begin{defin}[The Fast Evacuation Problem]
 Given a unit disk and two robots atarting at its center (the former with maximum speed $s \ge 1$ and the latter with maximum speed $1$),
 provide an algorithm such that \emph{both} robots reach an unknown exit lying on a boundary point of the disk.
 The two robots, namely \emph{Fast} and \emph{Slow}, are allowed to move within the entire unit disk, can only identify the exit when they stand on it,
 and can communicate wirelessly at any time.
\end{defin}

\begin{defin}
 An ``evacuation strategy'' is an algorithm on how each robot moves such that both robots have evacuated the disk at the end of its execution.
\end{defin}

The following remark is a direct consequence of the geometric environment in which this fast evacuation scenario takes place.

\begin{remark}
In any evacuation strategy, when either robot discovers the exit, the optimal strategy of the other one immediately reduces to following a beeline to the exit. 
\end{remark}

We now proceed with identifying key aspects of potential strategies.

\begin{defin}
 A ``both-explore'' strategy is a strategy for both robots to evacuate the disk, where (in the worst-case) both of them explore at least two distinct points on the boundary.
 We define the set of all both-explore strategies as $BES$.
\end{defin}

\begin{comment} CHECK THIS OUT with FES-like LB
Notice that for the set of $BES$-strategies, we are going to consider only values of $s < 2\pi+1$.
Otherwise, the fast robot can traverse the whole boundary before the slow one even touches it.
\end{comment}

\begin{defin}
 A ``fast-explores'' strategy is a strategy where \emph{only Fast} explores the boundary searching for the exit.
 Slow, eventually, only reaches the exit point and at any time it reaches no other point on the boundary of the disk.
 We define the set of all fast-explores strategies as $FES$.
\end{defin}

\begin{defin}
 A ``slow-explores'' strategy is a strategy where \emph{only Slow} explores the boundary searching for the exit.
 Fast, eventually, only reaches the exit point and at any time it reaches no other point on the boundary of the disk.
 We define the set of all fast-explores strategies as $SES$.
\end{defin}

Notice that for $s=1$, if only one robot explores the boundary, we randomly assign such a 
strategy to $FES$ or $SES$.

% Moreover, one can easily observe that an optimal $SES$ strategy for is one where Fast follows 
% slow up to an infinitesimally small distance always staying within the disk interior.
% Such a strategy yields an $\lim_{\epsilon\rightarrow 0}(1 + 2\pi - \epsilon) = 1 + 2\pi$ upper 
% bound, which is optimal given that Slow needs at least $1 + 2\pi$ time to explore the whole boundary.

Below, let $ALL$ stand for the set of all evacuating strategies.

\begin{prop}
 $(BES, FES, SES)$ forms a partition of $ALL$.
\end{prop}

\begin{proof}
 $BES \cap FES \cap SES = \emptyset$, since only Fast or only Slow or both explore the boundary.
 $ALL = BES \cup FES \cup SES$, since for any possible strategy at least one robot explores the boundary.
\end{proof}

We remark that, when considering $SES$ and $FES$ strategies, it can become a burden to forcefully 
keep the non-exploring robot away from the boundary.
E.g.\ if we only want Slow  to explore in an $SES$ strategy, the optimal behavior of Fast 
would be to mimic the behavior of Slow.
For $FES$ strategies with $s \leq 2$, it also proves to be most natural to allow Slow to move 
on the boundary, but to ignore it when Slow finds the exit first. 
For this reason we use $FES$ and $SES$ strategies in this sense.  
Alternatively, one could also let the non-exploring robot to move $\varepsilon$-close to the boundary. 

We do not consider $SES$ strategies in our analysis. An optimal $SES$ strategy is obviously to go to the boundary and explore the boundary clockwise or counterclockwise.
The worst case time is $1+2\pi$.

% when $s >1$, an optimal SES strategy can be replaced by a BES strategy by 
% having Fast mimic the behavior of Slow.  

\section{Upper Bounds}\label{sec:upper}

%We provide a $BES$, $FES$ and $SES$ strategy.
%SOMETHING MORE HERE

\subsection{The Half-Chord Strategy}
The idea for this strategy stems from the proof of the $FES$ lower bound to follow.
The worst-case analysis is performed for $s \in [2, \infty)$.
For the strategy details below, please refer to \autoref{fig:chord}.
Fast's trajectory is given in blue, while Slow's in red.
All arcs are considered in counterclockwise order.
\paragraph{\bf The Strategy.}

\emph{Fast} moves as follows until the exit is found:
\begin{itemize}
 \item for $t \in \left[0, \frac{1}{s}\right]$: moves toward $B$ and
 \item for $t \in \left(\frac{1}{s}, \frac{1+2\pi}{s}\right]$: traverses the boundary 
counterclockwise.
\end{itemize}
\emph{Slow} moves as follows until the exit is found:
\begin{itemize}
 \item \underline{Phase I}: for $t \in \left[0, \frac{2}{s}\right]$ moves toward $C$,
 \item \underline{Phase II}: for $t \in \left[\frac{2}{s}, \frac{1+2\arccos(-2/s)}{s}\right]$ moves toward $M$ via $\wideparen{CM}$ on disk $\left(O, \frac{2}{s}\right)$,
 \item \underline{Phase III}: for $t \in \left[\frac{1+2\arccos(-2/s)}{s}, \frac{1+2\pi}{s}\right]$ moves toward $B$ via the \emph{MB} segment.
\end{itemize}

In \autoref{tab:halfchord}, we shortly outline some core measurements on the emerging shape, e.g.\ angles and lengths, which will be useful in the proofs to follow.
We now continue with some useful propositions.

\begin{table}
\begin{center}
\begin{tabular}{|l||l|}
  \hline \hline
  $|OC| = \frac{2}{s}$ & by choice \\ \hline 
  $\wideparen{BA} = 2\arccos\left(-\frac{2}{s}\right)$ & by choice \\ \hline
  $\phi = \measuredangle BOC = \pi + 1/2$ & by choice \\ \hline
  $|\wideparen{CM}| = \frac{1}{s}(2\arccos\left(-\frac{2}{s}\right) - 1)$ & slow  on $M$ exactly when fast on $A$ \\ \hline
  $\theta = \measuredangle COM = \frac{s}{2}|\wideparen{CM}| = \arccos\left(-\frac{2}{s}\right) - 1/2$ & arc-to-angle \\ \hline
  $ \psi = \measuredangle MOB = 2\pi - \phi - \theta = \pi - \arccos\left(-\frac{2}{s}\right)$ & sum of angles around  $O$ \\ \hline
  $|AB| = 2\sin\left(2\arccos\left(-\frac{2}{s}\right)/2\right) = 2\sqrt{1 - \frac{4}{s^2}}$ & arc-to-chord computation \\ \hline 
  $|AM| = |MB| = |AB|/2 = \sqrt{1 - \frac{4}{s^2}}$ & since $M$ is the middle of the chord \\ \hline
  $\measuredangle OMB = \pi/2$ & perpendicular bisector through center \\  \hline
\end{tabular}
\end{center}
\caption{Measurements for Half-Chord Strategy}
\label{tab:halfchord}
\end{table}

 \begin{comment}
 \item $|OC| = \frac{2}{s}$ (by choice)
 \item $\wideparen{BA} = 2\arccos(-2/s)$ (by choice - in counterclockwise fashion)
 \item $\phi = \measuredangle BOC = \pi + 1/2$ (by choice)
 \item $|\wideparen{CM}| = \frac{1}{s}(2\arccos(-2/s) - 1)$ (such that slow arrives on $M$ exactly when fast arrives on $A$)
 \item $\theta = \measuredangle COM = \frac{s}{2}|\wideparen{CM}| = \arccos(-2/s) - 1/2$
 \item $ \psi = \measuredangle MOB = 2\pi - \measuredangle BOC -\measuredangle COM = \pi - \arccos(-2/s)$
 \item $|AB| = 2\sin(2\arccos(-2/s)/2) = 2\sin(\arccos(-2/s)) = 2\sqrt{1 - \frac{4}{s^2}}$ \\ (from arc-to-chord computation and the fact that $\sin(\arccos(x)) = \sqrt{1 - x^2}$ for any $x \ge 0$)
 \item $|AM| = |MB| = |AB|/2 = \sqrt{1 - \frac{4}{s^2}}$, since $M$ is the middle of the chord
 \item $\measuredangle OMB = \pi/2$, since the perpendicular bisector of any chord passes though the center (reference???)
\end{comment}

\begin{prop}
  Fast reaches $A$ exactly when Slow reaches $M$.
\end{prop}
\begin{proof}
 Fast reaches $A$ after $\frac{1+ 2\arccos(-2/s)}{s}$ time, 
 since it takes $\frac{1}{s}$ time for it to traverse \emph{OB} and $\frac{2\arccos(-2/s)}{s}$ time to traverse $\wideparen{BA}$.
 Slow reaches $C$ after time $\frac{2}{s}$.
 Then, it traverses $\wideparen{CM}$ for another $\frac{1}{s}(2\arccos(-2/s) - 1)$ time for a total of $\frac{1+ 2\arccos(-2/s)}{s}$.
\end{proof}

\begin{prop}
  Fast explores the whole boundary before Slow reaches $B$.
\end{prop}
\begin{proof}
 Slow reaches $M$ after $\frac{1+ 2\arccos(-2/s)}{s}$ time and then has to traverse $MB$ for another $\sqrt{1-\frac{4}{s^2}}$.
 Meanwhile, after $\frac{1+ 2\arccos(-2/s)}{s}$ time, Fast lies on $A$ and then has to traverse $\wideparen{AB}$ for another $\frac{2\pi-2\arccos(-2/s)}{s}$.
 It's adequate to see that $\sqrt{1-\frac{4}{s^2}} \ge \frac{2\pi-2\arccos(-2/s)}{s}$ for any $s \ge 2$.
% Consider $f(s) = \sqrt{1-\frac{4}{s^2}} - \frac{2\pi-2\arccos(-2/s)}{s} = \sqrt{1-\frac{4}{s^2}} - \frac{2\arccos(2/s)}{s}$.
% It suffices to notice that $f(2) = 0$ and that $\frac{df}{ds} = \frac{2\arccos(2/s)}{s^2} \ge 0$ for any $s \ge 2$.
\end{proof}

The aforementioned proposition, together with the fact that it takes $\frac{1+2\pi}{s}$ time for Fast to explore the whole boundary, provides us with the endtime for Phase III and the strategy in general.

The main result of this section follows from the combination of the upper bounds proved for Phase I, II, and III in the following subsections.

\begin{theorem}
\label{thm:fes_ub} For any $s \ge 2$, the worst-case evacuation time of the Half-Chord strategy is at most
 $\frac{1+2\arccos\left(-\frac{2}{s}\right)}{s} + \sqrt{1 - \frac{4}{s^2}}$.
\end{theorem}

% ---------------------- THE STRATEGY FIGURE --------------------------
\begin{figure} % example for s = 4
\includegraphics{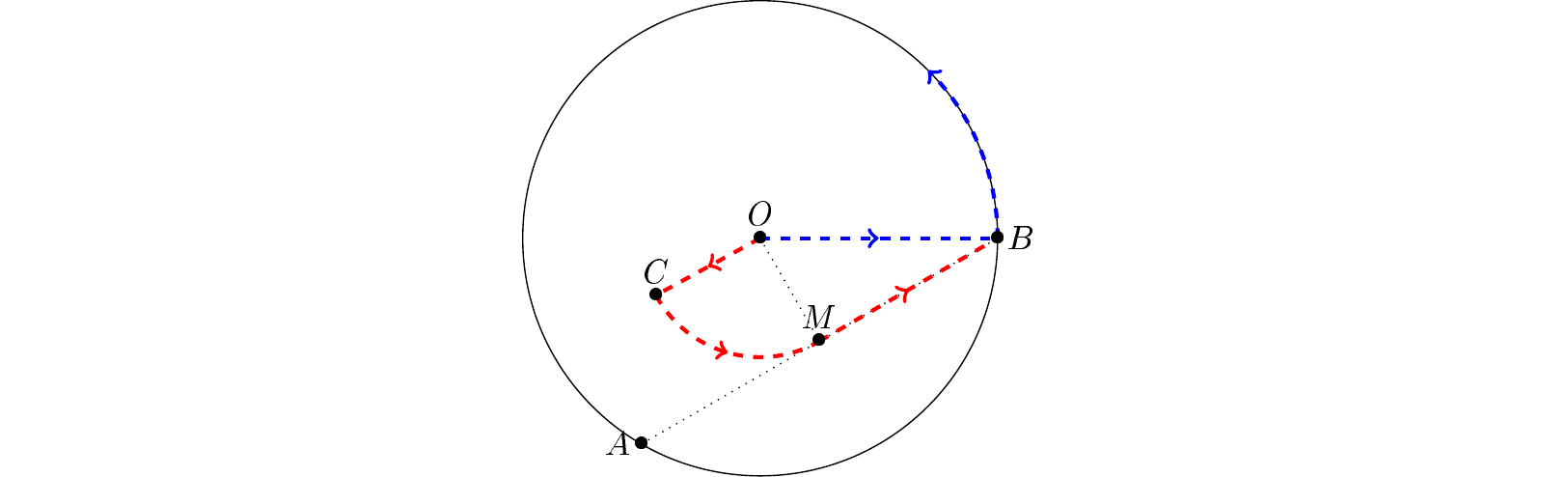}
\caption{The Half-Chord Strategy (Example for $s = 4$)}
\label{fig:chord}
\end{figure}

% ---------------------------------------------------------------------

\subsubsection{Phase I}
\begin{lemma}\label{thm:I}
 The Half-Chord evacuation strategy takes at most 
 $\frac{1+2\arccos\left(-\frac{2}{s}\right)}{s} + \sqrt{1 - \frac{4}{s^2}}$
 evacuation time, if the exit is found during Phase I.
\end{lemma}
\begin{proof}
We need only care about the time $t \in \left[\frac{1}{s}, \frac{2}{s}\right]$, since for less time Slow has not yet reached the boundary. 
Imagine that the exit is discovered after $\frac{1+a}{s}$ time (for $a \in [0, 1]$).
For a visualization, the reader can refer to \autoref{fig:phaseI}.
Slow has covered $\frac{1+a}{s}$ distance on the \emph{OC} segment, while Fast has explored an $a$ part of $\wideparen{BA}$.
Slow now takes a segment from its current position (namely $D$) to the exit $E$.
To compute $|DE|$ we use the law of cosines in $\triangle DOE$.
Let $\omega = \measuredangle DOE$.
In case $a \le \frac{1}{2}$, $\omega \le \pi$, and more accurately $\omega = a + \psi + \theta = \pi + a - \frac{1}{2}$.
In case $a > \frac{1}{2}$, $\omega > \pi$, and more accurately $\omega = 2\pi - a - \psi - \theta$.
Since $\cos(2\pi - x) = \cos(x)$, we  can consider the two cases together.
We compute, 
$ |DE| = \sqrt{ |OE|^2 + |OD|^2 - 2|OE||OD|\cos(\omega)} =  \sqrt{1 + \frac{(1+a)^2}{s^2} - 2\frac{1+a}{s}\cos(\pi+ a - 1/2)} =
 \sqrt{1 + \frac{(1+a)^2}{s^2} + 2\frac{1+a}{s}\cos(1/2 - a)}$.
Overall, the worst-case evacuation time is given by
$ \max_{a \in [0, 1]} \left\{\frac{1+a}{s} + \sqrt{1 + \frac{(1+a)^2}{s^2} + 2\frac{1+a}{s}\cos(1/2 - a)}\right\} $.
To conclude the proof, it suffices to observe that $\frac{2}{s} + \sqrt{1 + \frac{2^2}{s^2} + 2\frac{2}{s}}$ is an upper bound to the above quantity, 
since $a \le 1$ and $\cos(\cdot) \le 1$. Finally, $\frac{2}{s} + \sqrt{1 + \frac{2^2}{s^2} + 2\frac{2}{s}} \leq \frac{1+2\arccos\left(-\frac{2}{s}\right)}{s} + \sqrt{1 - \frac{4}{s^2}}$
for any $s \geq 2$.
%Let $lb(s) = \frac{1+2\arccos\left(-\frac{2}{s}\right)}{s} + \sqrt{1 - \frac{4}{s^2}}$.
%We show that $lb(s) - f(s) \ge 0$ for any $s \ge 2$.
%For $s = 2$, $lb(2) - f(2) = \frac{1 + 2\pi}{2} - 1 - \sqrt{1 + 1 + 2} = \pi - \frac{5}{2} \ge 0$.
%Finally, $\frac{d(lb-f)(s)}{ds} = \frac{3 - 2\arccos(-2/s)}{s^2} < 0$ for any $s \ge 2$ and 
%$\lim_{s \rightarrow \infty} (lb(s) - f(s)) = 0$.
\end{proof}

% ------------------ PHASE I & II FIGURE -----------------
\begin{figure} % example for s = 4
\begin{center}
\begin{subfigure}{.45\textwidth} % PHASE I
 \includegraphics{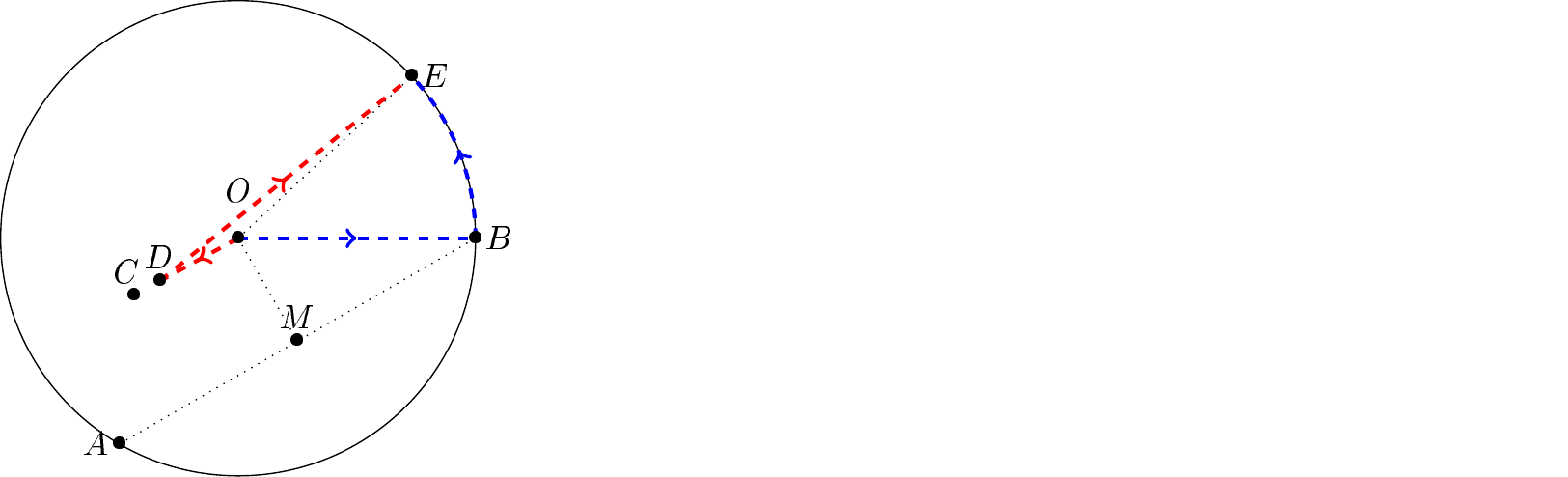}
\subcaption{Exit during Phase I ($a = 0.75$)}
\label{fig:phaseI}
\end{subfigure}
\begin{subfigure}{.45\textwidth} % PHASE II
  \includegraphics{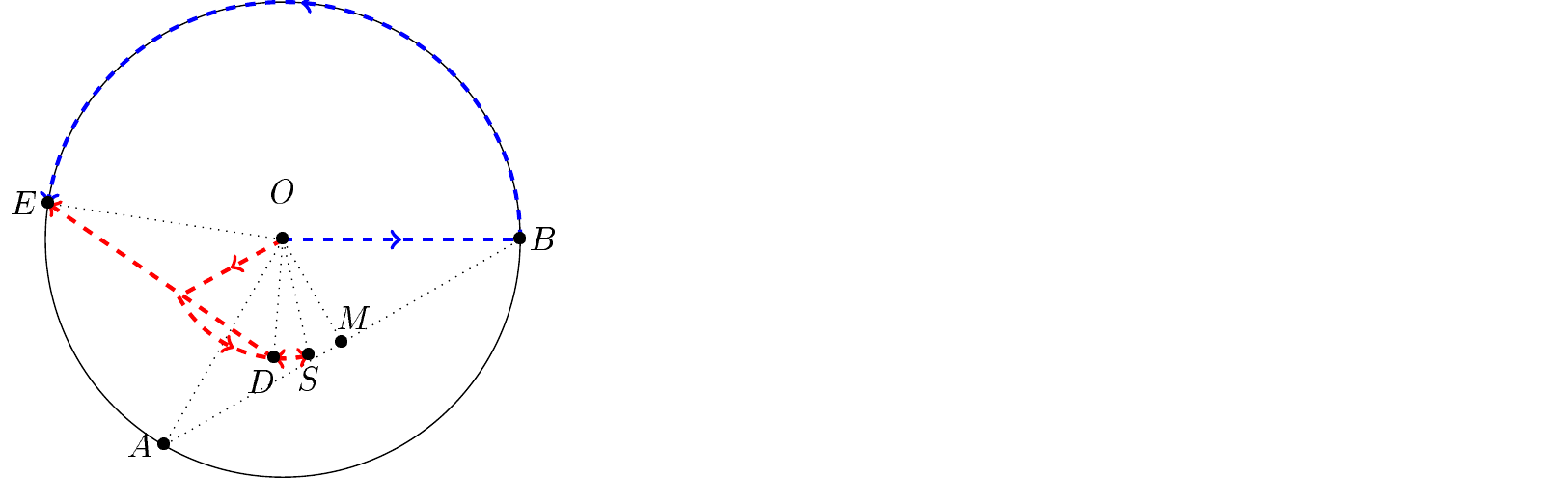}
  \subcaption{Exit during Phase II ($\tau = 0.3$)}
\label{fig:phaseII}
\end{subfigure}
\end{center}
\caption{Exit during Phase I \& II (Examples for $s = 4$)} 
\label{fig:phaseI&II}
\end{figure}
% ------------------------------------------------------

\subsubsection{Phase II}
\begin{lemma}\label{thm:II}
 The Half-Chord evacuation strategy takes at most 
 $\frac{1+2\arccos\left(-\frac{2}{s}\right)}{s} + \sqrt{1 - \frac{4}{s^2}}$
 evacuation time, if the exit is found during Phase II.
\end{lemma}

\begin{proof}
 We prove that the worst-case placement for the exit is point $A$.
 Suppose the exit $E$ is found at the time when Slow lies on point $S$ and has not yet covered a $\tau$ part of $\wideparen{CM}$.
 The corresponding central angle is $\frac{s\tau}{2}$, since $\wideparen{CM}$ is an arc on $(O, \frac{2}{s})$.
 At the same time, Fast has not yet explored an $s\tau$ part of $\wideparen{BA}$ with a corresponding central angle of size $s\tau$.
 Then, Slow can move backwards on the boundary of $(O, \frac{2}{s})$ for another $\tau$ distance to point $D$.
 Now, the central angle from $D$ to $M$ is $\frac{s\tau}{2} + \frac{s\tau}{2}  = s\tau$ and matches the central angle between $E$ and $A$.
 Thence, due to shifting by the same central angle, we get $\measuredangle EOD = \measuredangle EOA + \measuredangle AOD = \measuredangle DOM + \measuredangle AOD = \measuredangle AOM$.
 Moreover, since $|OD| = |OM| = \frac{2}{s}$ and $|OE| = |OA| = 1$, triangles $\triangle EOD$ and $\triangle AOM$ are congruent
 meaning that $|ED| = |AB|$.
 To sum up, if the exit is discovered $\tau$ time before Slow reaches $M$, it takes at most another $\tau + \sqrt{1 - \frac{4}{s^2}}$ time for it to reach it.
 At the same time, it would take $\tau + \sqrt{1 - \frac{4}{s^2}}$ for it to reach $A$.
 Hence, exiting through $A$ is the worst-case scenario and yields a total time of $\frac{1+2\arccos\left(-\frac{2}{s}\right)}{s} + \sqrt{1 - \frac{4}{s^2}}$.
\end{proof}

\subsubsection{Phase III}
\begin{lemma}\label{thm:III}
 The Half-Chord evacuation strategy takes at most 
 $\frac{1+2\arccos\left(-\frac{2}{s}\right)}{s} + \sqrt{1 - \frac{4}{s^2}}$
 evacuation time, if the exit is found during Phase III.
\end{lemma}

\begin{proof}
 Since $\frac{1+2\arccos\left(-\frac{2}{s}\right)}{s}$ time has already passed at the beginning of Phase III,
 it suffices to show that at most $\sqrt{1-\frac{4}{s^2}}$ time goes by when the exit is discovered within $\wideparen{AB}$.

Suppose that the exit is discovered $\tau$ time units after the beginning of Phase III.
Then, Slow lies at $C$ (\autoref{fig:phaseIII}), $\tau$ distance away from $M$ on the $MB$ segment.
On the other hand, Fast lies on $E$, an $s\tau$ distance away from $A$ on $\wideparen{AB}$.

Consider a disk with center $C$ and radius $r = \sqrt{1-\frac{4}{s^2}} - \tau$.
One can notice that $(C, r)$ intersects $(O, 1)$ at two points: one of them is $B$ and the other one is $D$, where $D$ is included in $\wideparen{AB}$, since $|AC| \ge r$ for any choice of $\tau \ge 0$.
Moreover, we draw the chord $DB$ and its middle point, say $M'$.
Now, notice that $OM'$ is perpendicular to $DB$, since $DB$ is a chord of $(O, 1)$ and also that $OM'$ passes through $C$, since $DB$ is also a chord of $(C, r)$.
To conclude, we exhibit that $E$ is included in $\wideparen{DB}$.
Equivalently, that $|\wideparen{AE}| \ge |\wideparen{AD}|$.
We look into two cases.

First, that $\measuredangle AOD \le \measuredangle AOM$.
In this case, we compute $\measuredangle AOD = \measuredangle AOM - \measuredangle DOM = \measuredangle MOB - \measuredangle DOM = \measuredangle MOM' + \measuredangle M'OB - \measuredangle DOM = \measuredangle MOM' + \measuredangle DOM' - \measuredangle DOM = 2 \cdot \measuredangle MOM'$, 
since $\measuredangle AOM = \measuredangle MOB$ and $\measuredangle M'OB = \measuredangle DOM'$ from the fact that $OM$ ($OM'$) bisects $AB$ ($DB$). Moreover, $\measuredangle DOM' - \measuredangle DOM =  \measuredangle MOM'$.
We compute $\measuredangle MOM' = \arctan(s\tau/2)$ by the right triangle $\triangle MOC$.
Finally, $\measuredangle AOD = 2\arctan(s\tau/2) \le s\tau = \measuredangle AOE$, since $\arctan(x) \le x$ for $x \ge 0$.

For the second case, $\measuredangle AOD > \measuredangle AOM$.
Then, $\measuredangle AOD = \measuredangle AOM + \measuredangle MOD =  \measuredangle MOB + \measuredangle MOD = 
\measuredangle MOM' + \measuredangle M'OB + \measuredangle MOD = \measuredangle MOM' + \measuredangle DOM' + \measuredangle MOD = 2\cdot \measuredangle MOM'$, again by using the equalities deriving from bisecting the chords.
The rest of the proof follows as before.
\end{proof}

% ------------------ PHASE III FIGURE -----------------
\begin{figure} % example for s = 4
  \begin{center}
   \begin{subfigure}{.45\textwidth} % first case
     \includegraphics{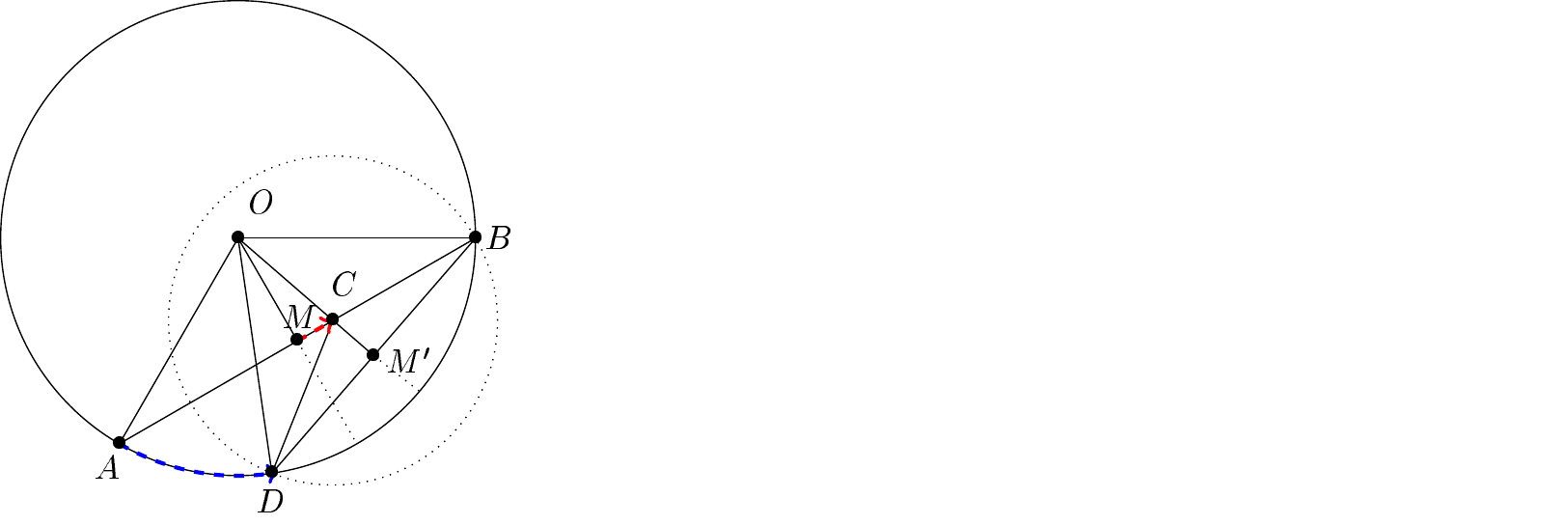}
 \subcaption{First case ($\tau = \frac{1}{5}$)}
 \label{fig:IIIi}
\end{subfigure}
\begin{subfigure}{.45\textwidth} % second case
     \includegraphics{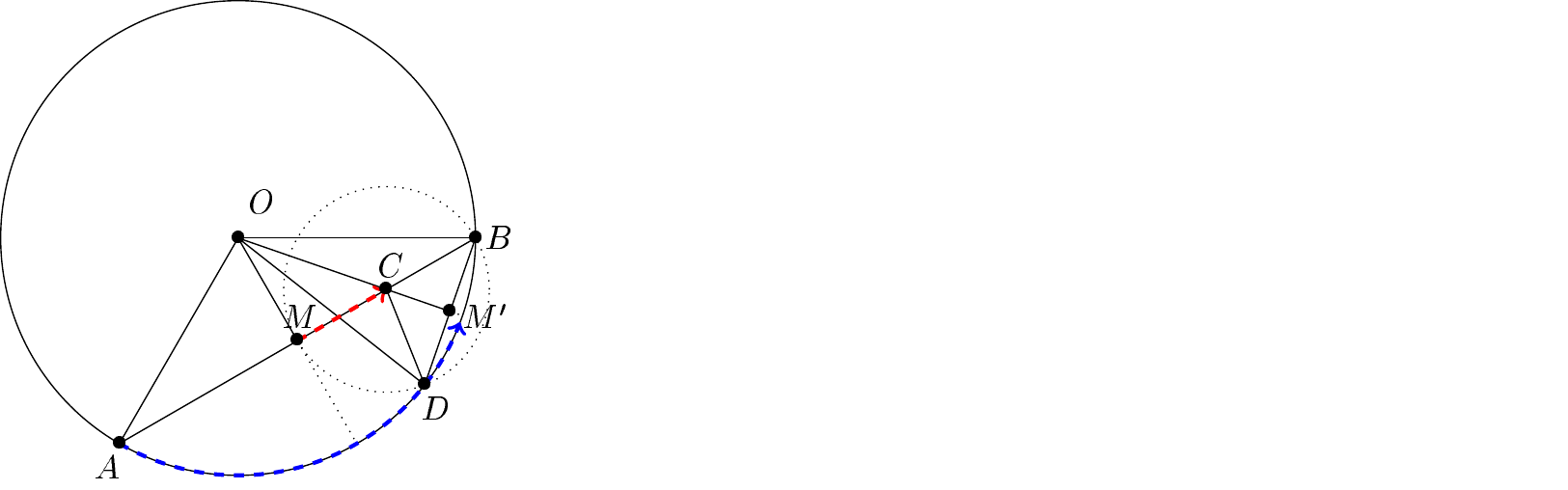}
 \subcaption{Second case ($\tau = \frac{1}{2}$)}
 \label{fig:IIIii}
\end{subfigure}
\end{center}
\caption{Exit during Phase III \\ (Example for $s = 4$; The exit $E$ lies at the end of the blue arrow)}
\label{fig:phaseIII}
\end{figure}
% ------------------------------------------------------

\begin{comment}
\subsection{$s \in (1.856, 2)$: The Slow Half-Chord Strategy}
 For the remaining interval of possible speeds for the fast robot, it suffices to consider 
 a \emph{slow} version of the Half-Chord Strategy for the case $s=2$.
 That is, the two robots follow the same trajectories as in the $s=2$ Half-Chord strategy.
 But, in this case, the fast robot is moving with speed $s \in (1.856, 2)$, while the slow 
 one is moving with speed $\frac{s}{2} \le 1$.
 The ratio between the two speeds remains the same (i.e. $2$), distances are preserved and 
 hence the $\frac{1+2\pi}{s}$ upper bound also applies.
 
 \begin{theorem}
  For any $s \in (1.856, 2)$, the worst-case evacuation time of the Slow Half-Chord strategy 
  is at most $\frac{1+2\pi}{s}$, which is optimal.
 \end{theorem}
 \begin{proof}
  The two robots follow the exact same trajectories as in the $s=2$ Half-Chord strategy only 
  with a time delay factor of $\frac{2}{s}$.
  Thence, the worst-case evacuation time becomes $\frac{2}{s} \cdot \frac{1+2\pi}{2} = \frac{1+2\pi}{s}$.
  Optimality follows from \autoref{thm:lb}.
 \end{proof}
\end{comment}

\subsection{The Half-Chord Strategy for $1 \leq s \leq 2$}\label{sec:slow-hc}
We first observe that, for $s=2$, the name ``Half-Chord'' is slightly misleading, as the points $A$, $B$, and $M$ coincide.
The time needed for $s = 2$ is, as shown in Theorem \ref{thm:fes_ub}, $\frac{1+2\pi}{s}$.
Note also that the Half-Chord strategy is a BES strategy for $s=2$.

For $s<2$, Slow can simply move even slower, namely with speed $\frac{s}{2}$.
Using the same paths as for $s=2$, this provides the same upper bound of  $\frac{1+2\pi}{s}$.

\begin{theorem}
For $1 \leq s \leq 2$, the (generalized) Half-Chord strategy leads to a $\frac{1+2\pi}{s}$ evacuation time.
\end{theorem}

\subsection{The Both-to-the-Same-Point Strategy}
This $BES$ strategy follows the same key idea presented in \cite{s1} where proven to be 
optimal for $s = 1$.

\begin{paragraph}{The Strategy.}
In the \emph{Both-to-the-Same-Point Strategy} (shortly \emph{BSP} strategy), 
initially both robots set out toward the same boundary point moving in a beeline.
Once they arrive there, they move to opposite directions along the boundary.
Without loss of generality, Fast moves counterclockwise along the boundary, while Slow moves clockwise.
This goes on, until the exit has been found by either robot or the robots meet each other on the boundary. 
For a visualization of the strategy, see \autoref{fig:bsp}.
Fast's trajectory is given in blue, while Slow's in red.
\end{paragraph}

\begin{figure} 
\begin{center}
 \begin{subfigure}{.45\textwidth} % THE BSP STRATEGY
 \begin{center}
  \includegraphics{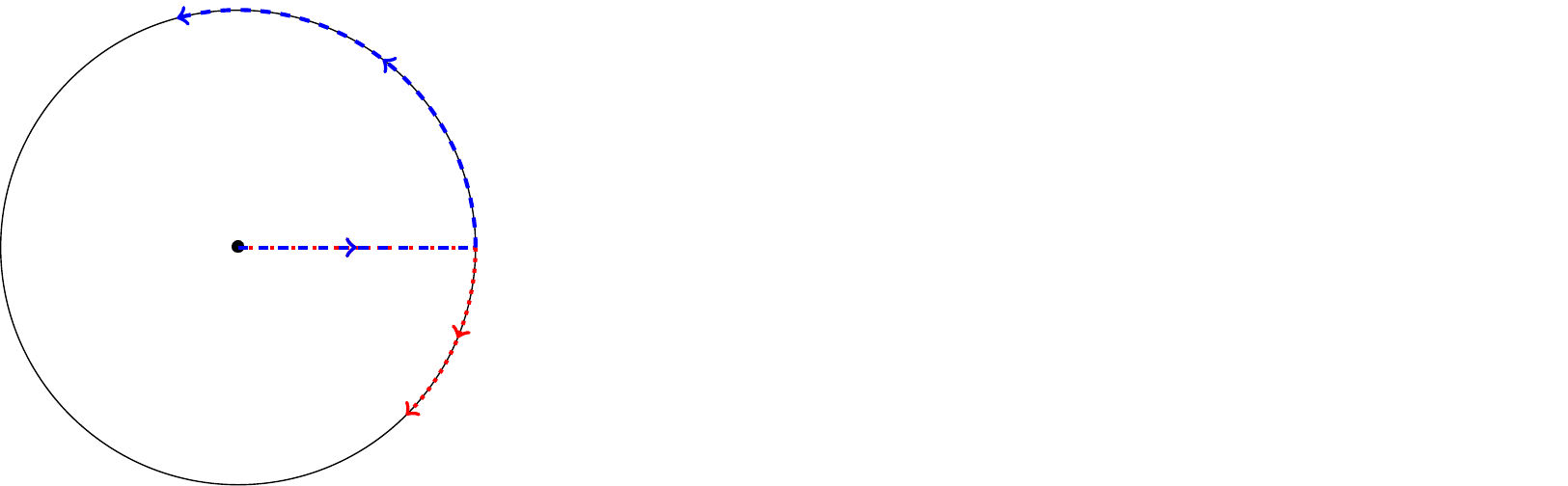}
 \end{center}
 \subcaption{The $BSP$ Strategy}
 \label{fig:bsp}
 \end{subfigure}
 \begin{subfigure}{.45\textwidth}	% EXIT BEFORE SLOW EXPLORES
 \begin{center}
 \includegraphics{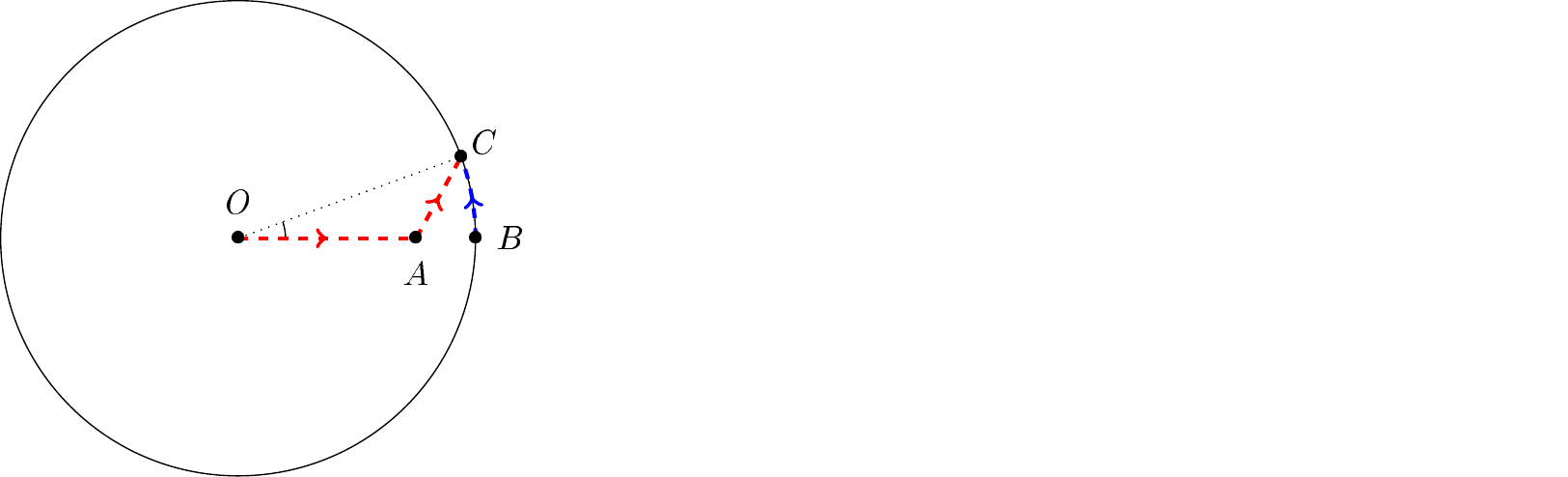}
\end{center}
\subcaption{Exit before Slow reaches the boundary \\ (Example for $s = 1.8$ and $a = 0.35$)}
\label{fig:cos}
\end{subfigure}
\end{center}
\caption{The $BSP$ Strategy and an Evacuation Example}
\label{fig:BSP}
\end{figure}

Below, we restrict the analysis of the BSP strategy only for $s \in [1, 2]$, since for $s > 2$ the strategy presented in the previous section yields a stronger upper bound.
The rest of the section is devoted to proving the main theorem.

\begin{theorem}\label{thm:bsp}
The BSP strategy requires evacuation time at most $1+ 2\sqrt{1 - \frac{1}{(s+1)^2}} + \frac{2\arccos(-\frac{1}{s+1}) -s + 1}{s+1}$ when $s \in [1,2]$.
\end{theorem}

\subsubsection{Exit found before Slow reaches the boundary}
\begin{lemma}\label{thm:bsp1}
 It takes at most $1 + \sqrt{2 - 2\cos(s-1)}$ time (where $s \in [1, 2]$) for both robots to evacuate in the $BSP$ strategy, 
when the exit is found before the slow robot has reached the boundary.
\end{lemma}

\begin{proof}
Let $a$ stand for the distance Fast has explored on the boundary before finding the exit.
Notice that $a \le s-1 \le 1$, since $a$ stands for a covered distance before Slow reaches the boundary.
The total evacuation time is  the time needed for Fast to find the exit and then for Slow to reach it.
Let $b$ stand for the latter.
Then, the worst-case evacuation time is
$ \max_{0\le a \le s-1} \left\{\frac{a+1}{s} + b\right\}$
, where 
$ b = \sqrt{1 + \left(\frac{a+1}{s}\right)^2 - 2\cdot\frac{a+1}{s}\cos(a)} $
by the cosine law in the formed triangle ($\triangle OAC$ in \autoref{fig:cos} with $|OC| = 1$, $|OA| = \frac{1+a}{s}$ and $\measuredangle AOC = a$).
Let $f(a,s) = \frac{a+1}{s} + b$.
Then, $\frac{\partial}{\partial a}f(a) \ge 0$ for any $a \le s-1$.
Consequently, $f(a,s)$ is a non-decreasing function of $a$ in this interval meaning that the maximum is attained on $a = s-1$.
This results to a worst-case time of $1 + \sqrt{2 - 2\cos(s-1)}$.
\end{proof}

\begin{comment}
% -------------------------------------------------
\begin{figure} % example for s = 1.8, a = 0.35
\begin{center}
 \begin{tikzpicture}[scale = 2.5]
 
 % auxiliary lines
 \draw (0,0) circle (1 cm); 
 \draw[dotted] (0,0) -> (0.9394, 0.3429);
 \draw (0.2,0) arc (0:deg(0.35):0.2); % angle
 
 \begin{scope}[very thick,decoration={
    markings,
    mark=at position 0.5 with {\arrow{>}}}
    ] 
  % slow
  \draw[red, dashed, postaction={decorate}] (0, 0) ->  (1.35/1.8, 0);
  \draw[red, dashed, postaction={decorate}] (1.35/1.8, 0) -> (0.9394, 0.3429);
  % fast
%  \draw[blue, dashed, postaction={decorate}] (0,0) -> (1,0);
  \draw[blue, dashed, postaction={decorate}] (1,0) arc (0: deg(0.35): 1);
 \end{scope}
 % points
 \foreach \Point in {(0, 0), (1,0), (1.35/1.8, 0), (0.9394, 0.3429)}{
    \node at \Point {\textbullet};
}
 % labels
 \node (O) at (0,0.15) {$O$};
 \node (B) at (1.15, 0) {$B$};
 \node (A) at (1.35/1.8, 0-0.15) {$A$};
 \node (C) at (0.9394+0.1, 0.3429+0.06) {$C$};
 \end{tikzpicture}
\end{center}
\caption{Exit before slow explores the boundary (Example for $s = 1.8$ and $a = 0.35$)}
\label{fig:cos}
\end{figure}
% ---------------------------------------------------------------------
\end{comment}

\subsubsection{Exit found after Slow has explored a part of the boundary}
\begin{lemma}\label{thm:bsp2}
 In the $BSP$ strategy (where $s \in [1, 2]$), when the exit is found after Slow has explored some part of the boundary, the evacuation time is at most
 \begin{itemize}
  \item  $\frac{2s + \pi + 4}{s+1}$, when the angle between the two robots is less or equal to $\pi$ and
  \item  $1+ 2\sqrt{1 - \frac{1}{(s+1)^2}} + \frac{2\arccos(\frac{1}{-s-1}) -s + 1}{s+1}$ when the angle is between $\pi$ and $2\pi$.
 \end{itemize}
\end{lemma}

\begin{proof}
Let $d$ stand for the distance Fast has covered on the boundary to find the exit counted only after Slow has started exploring.
Using this notation, the explored part of the boundary is a function of $d, s$, namely $angle(d, s)  = s-1 + d + \frac{d}{s} = s - 1 + d(1+\frac{1}{s})$, since Slow explores distance $\frac{d}{s}$, while Fast explores distance $d$, and an $s-1$ part has already been covered.
The name $angle(\cdot, \cdot)$ is chosen, since the quantity also represents the angle between the robots from the center of the unit disk.
We break the analysis into two cases:
\begin{itemize}
 \item $angle(d, s) \le \pi$: \\
	In this case, $s - 1 + d(1+\frac{1}{s}) \le \pi$, which results to $d \le \frac{\pi-s+1}{1+1/s}$.
	Notice that the bound is $\ge 0$ for $s \in [1, \pi+1]$.
	The worst-case evacuation time is given by computing the function 
	$\max_{0 \le d \le \frac{\pi-s+1}{1+1/s}} \left\{1 + \frac{d}{s} + 2\sin\left( \frac{d(1+\frac{1}{s}) + s  -1}{2}\right)\right\}$,
	where the last addend accounts for the chord length needed to be covered by Slow.
	We denote $g(d,s)$ the function to be maximized.
	Similarly to before, we can see that $\frac{\partial}{\partial d} g(d, s) \ge 0$ for any choice of $s \in [1, 2]$ and any $d \in [0,\frac{\pi-s+1}{1+1/s}]$.
	Hence, the maximum is attained at $d = \frac{\pi-s+1}{1+1/s}$ for a worst-case time of $\frac{2s + \pi + 4}{s+1}$.

\item $\pi < angle(d, s) < 2\pi$: \\
	In this case,  $d \in (\frac{\pi-s+1}{1+1/s}, \frac{2\pi-s+1}{1+1/s})$.
        The function to be maximized is again $g(d,s)$.
        The maximum is attained for $d' = \frac{2\cdot s \cdot \arccos(-1/(s+1))-s+1}{s+1}$ yielding an upper bound of
        $1+ 2\sqrt{1 - \frac{1}{(s+1)^2}} + \frac{2\arccos(\frac{1}{-s-1}) -s + 1}{s+1}$  for $1 \le s \le 2$.
\end{itemize}

Finally, notice that we need not care about the case where Slow finds the exit, since the time taken for Fast to traverse the same chord will be less than the worst-case examined.
\end{proof}

%%%%%%%%%%%%%%%%%%%%%%%%%%% [THE NEW UPPER BOUND] %%%%%%%%%%%%%%%%%%%%%%%%%%%%%%%%%%%%%%%%%

\subsection{The Fast-Chord Strategy}\label{sec:new-ub}
In the Half-Chord strategy for $s=2$, we observe that the final point reached after Phase I, that is point $C$, lies on the disk boundary.
Thence, after that, Slow explores $\wideparen{CB}$, but so does Fast (since by its strategy it explores the whole boundary).
This seems like an unnecessary double-exploring of this part of the boundary.
Thus, we propose a new strategy, where Fast reaches $C$ as usual, but then traverses the $CB$ chord, instead of $\wideparen{CB}$.
Furthermore, we could variate the position of $C$, in order for Fast to reach $B$ (for the second time) exactly when Slow reaches $D$ (a point before $B$) and so get Fast to explore some part of the boundary in clockwise fashion as well.
In this case, Slow does not traverse the whole $\wideparen{CB}$.
Let us now describe more formally this \emph{Fast-Chord} family of strategies. 
All arcs are considered in \emph{counterclockwise} fashion unless otherwise stated.
In the description below, let $|\wideparen{BA}| = s -1$, $x_1 = |\wideparen{AC}|$, $x_2 = |CB|$, $x_3 = |\wideparen{DB}|$ and $y = |\wideparen{CB}|$.
For a pictorial representation, the reader can refer to \autoref{fig:fast-chord}.

\begin{paragraph}{The Strategy.}
 \emph{Fast} moves as follows until the exit is found:
\begin{itemize}
 \item for $t \in \left[0, \frac{1}{s}\right]$ moves toward $B$,
 \item \underline{Phase I}: for $t \in \left(\frac{1}{s}, 1\right]$  traverses $\wideparen{BA}$,
 \item \underline{Phase IIa}: for $t \in \left(1, 1 + \frac{x_1}{s}\right]$ traverses $\wideparen{AC}$,
 \item \underline{Phase IIb}: for $t \in \left(1 + \frac{x_1}{s}, 1 + \frac{x_1 + x_2}{s}\right]$ traverses $CB$ and
 \item \underline{Phase IIc}: for $t \in \left(1 + \frac{x_1 + x_2}{s}, 1 + \frac{x_1 + x_2}{s} + \frac{x_3}{s+1}\right]$ moves toward $D$ (clockwise) till it meets Slow.
\end{itemize}
\emph{Slow} moves as follows until the exit is found:
\begin{itemize}
 \item for $t \in [0, 1]$ moves toward $C$,
 \item for $t \in (1, 1+y]$ traverses $\wideparen{CD}$,
 \item for $t \in \left(1+y, 1+y+\frac{x_3}{s+1}\right]$ traverses $\wideparen{DB}$ till it meets Fast.
\end{itemize}
\end{paragraph}

% ---------------------- THE FAST-CHORD STRATEGY FIGURE --------------------------
\begin{figure}
\begin{center}
 \includegraphics{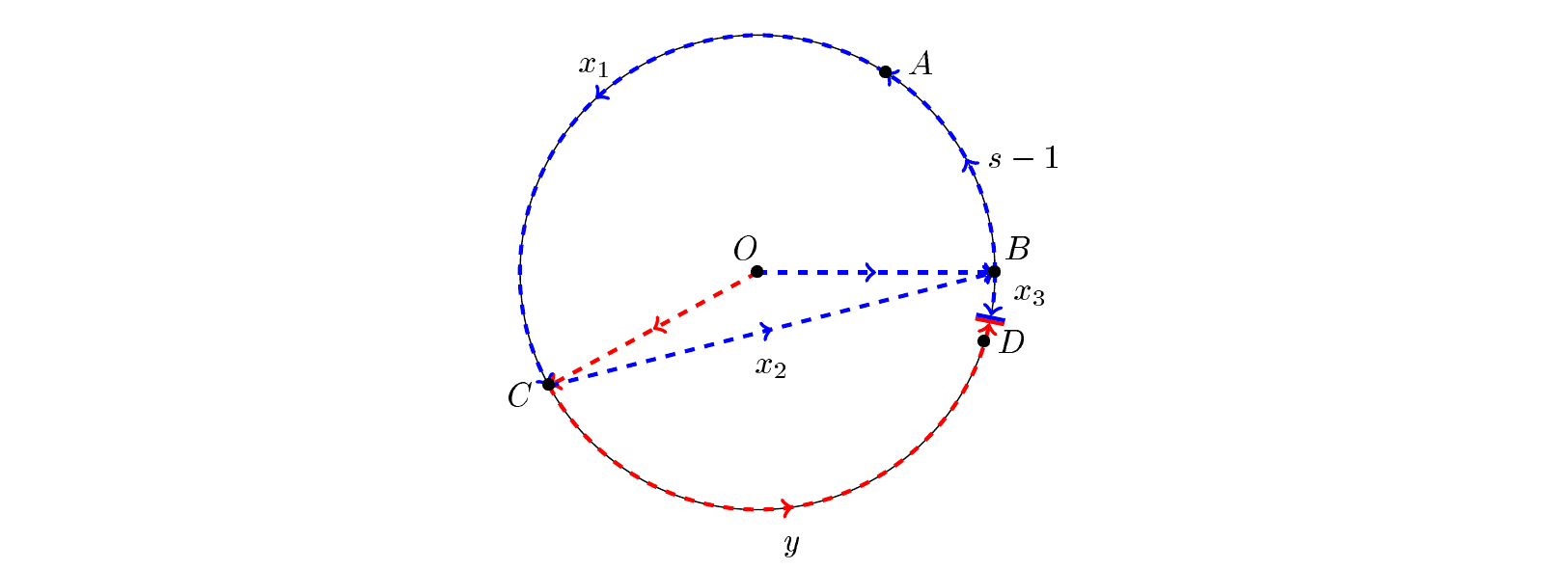}
\end{center}
\caption{The Fast-Chord Family of Strategies}
\label{fig:fast-chord}
\end{figure}

The following system of equations describes the relationship between the variable distances:

\begin{center}
$\left\{
\begin{array}{l l l}
  x_1 + y + x_3 + s - 1 &=  2\pi &\mbox{(I)}\\
  x_2 &=   2\sin\left(\frac{x_3+y}{2}\right)  &\mbox{(II)}\\
  x_1 + x_2 &= s\cdot y  &\mbox{(III)} 
\end{array}
\right.$
\end{center}

Equation (I) suggests how the disk boundary is partitioned.
Equation (II) suggests that $x_2$ is the chord of an  arc with length $x_3 + y$.
Equation (III) suggests that Fast traverses $x_1$ and $x_2$ at the same time as slow traverses $y$.
That is, since Fast lies on $A$ exactly when Slow lies on $C$, then Fast arrives at $B$ (for the second time) exactly when Slow arrives at $D$.
The latter happens at time $1+y = 1 + \frac{x_1+x_2}{s}$. The remaining $x_3$ part of the boundary can be explored in time $\frac{x_3}{s+1}$, since both robots explore it concurrently until they meet.
Hence, within $\frac{x_3}{s+1}$ time, they can explore a distance equal to $s \cdot \frac{x_3}{s+1} + \frac{x_3}{s+1} = (s+1)\cdot \frac{x_3}{s+1} = x_3$.
All variables are non-negative, since they represent distance.

The idea behind this paradigm is to try different values for $x_3$ and then solve the above system to extract $x_1, x_2$ and $y$.
Nonetheless, due to the $\sin(\cdot)$ function in equation (II), a symbolic solution is not possible to obtain.
Thence, we hereby provide bounds computed \emph{numerically}.
For any value of $s$, we iterate over all possible $x_3$ values and then solve the above system numerically.
For each $x_3$ value and for each exploration phase, we use a small time step and compute the worst-case evacuation time.
Then, we can select the $x_3$ value which minimizes this worst-case time.
All this numerical work is implemented in Matlab. 
We iterate over $x_3$ with a step of $10^{-2}$ in the interval $\left[0, 2\pi-s+1\right]$.
The upper bound for $x_3$ stems from the case $x_1 = y = 0$. 
Indeed, notice that for $s = 1$, Fast-Chord is exactly $BSP$, when we set $x_1 = y = 0$.
For the time parameter, namely $t$, we again use a step of $10^{-2}$ in the interval $\left[0, 1 + \frac{x_1+x_2}{s} + \frac{x_3}{s+1}\right]$.
We use a parametric representation of the disk (where the center $O$ lies on coordinates $(0,0)$) to calculate the \emph{Euclidean distance} between the two robots.
Below, let $Fast_x$ and $Fast_y$ stand for the $(x, y)$ coordinates of Fast's position and similarly $Slow_x$ and $Slow_y$ for Slow.
The distances between the two robots at any given time are as follows:

\begin{paragraph}{Phase I.}
 At time $t \in \left(\frac{1}{s}, 1\right]$, Fast has covered an $st - 1$ part of $\wideparen{BA}$ (until point $A'$),
 while Slow has covered a $t$ part of $OC$ (until point $C'$); see \autoref{fig:fcI}. 
 Their distance is given by applying the cosine law in $\triangle A'OC'$.
 We compute the \emph{in-triangle} angle $\measuredangle A'OC'$.
 In case that $\wideparen{A'C'} \le \pi$ (case i), then 
 $\measuredangle A'OC' = \wideparen{BC} - \wideparen{BA'} = s - 1 + x_1 - (st - 1) = s(1-t) + x_1$.
 Otherwise, if $\wideparen{A'C'} > \pi$ (case ii), then 
 $\measuredangle A'OC' = 2\pi - \wideparen{A'A} - \wideparen{AC} = 2\pi - (s - 1 - (st - 1)) - x_1 = 2\pi - s(1-t) - x_1$.
 In either case, $|A'C'| = \sqrt{|OA'|^2 + |OC'|^2 - 2|OA'||OC'|\cos(\measuredangle A'OC')} = \sqrt{1 + t^2 - 2t\cos(s(1-t)+x_1)}$, since $\cos(2\pi-x) = \cos(x)$ for any $x$.
\end{paragraph}

% Phase I figure
\begin{figure}

\begin{subfigure}{.45\textwidth}
\begin{center}
 \includegraphics{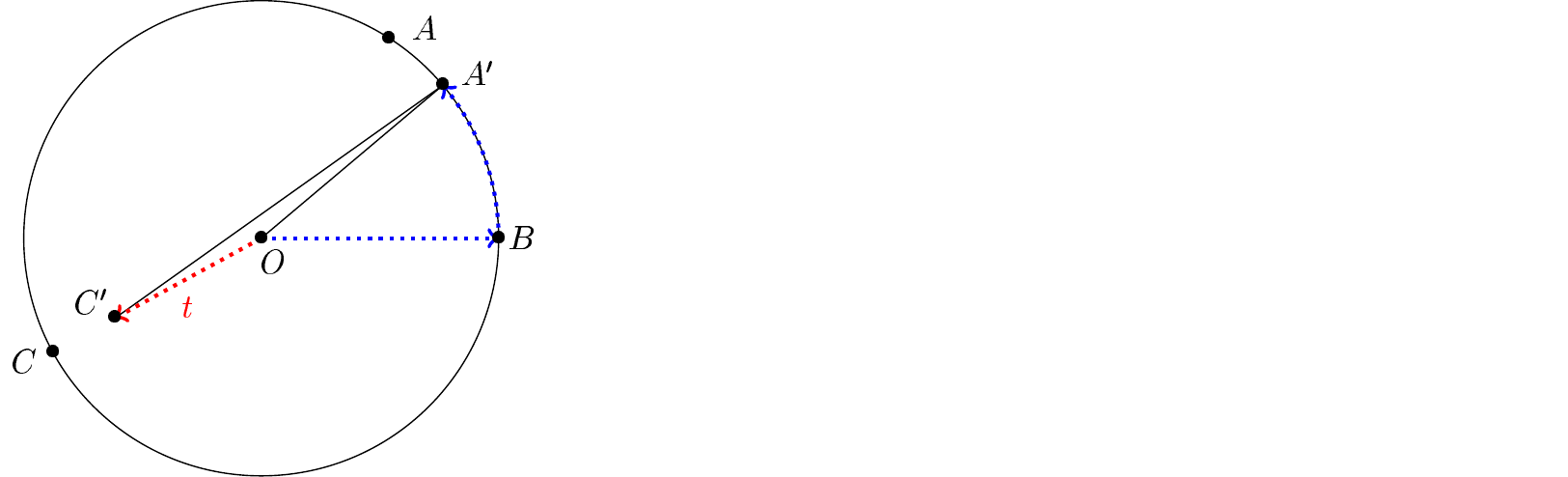}
\end{center}
\subcaption{Case (i)}
\label{fig:fcIi}
\end{subfigure}
\begin{subfigure}{.45\textwidth}
\begin{center}
 \includegraphics{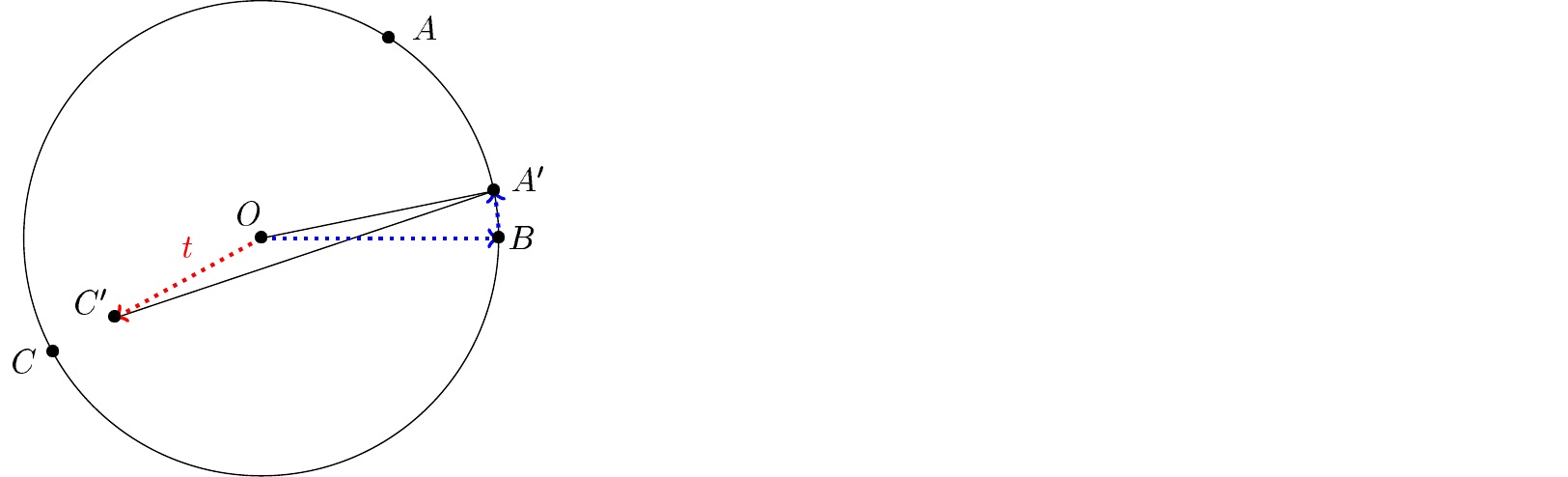}
\end{center}
\subcaption{Case (ii)}
\label{fig:fcIii}
\end{subfigure}
\caption{Fast-Chord: Phase I}
\label{fig:fcI}
\end{figure}

\begin{paragraph}{Phase IIa.}
 At time $t \in \left(1, 1+\frac{x_1}{s}\right]$, both robots are traversing their respective arcs in counterclockwise fashion.
 Their positions are the following:
 
 $$(Fast_x, Fast_y) = \left(\cos\left(s\left(t-\frac{1}{s}\right)\right), \sin\left(s\left(t-\frac{1}{s}\right)\right)\right)$$
 $$(Slow_x, Slow_y) = \left(\cos(s-1+x_1 + t - 1), \sin(s-1+x_1+t-1)\right)$$
 
 taking into account the initial timestep when they begin traversing their corresponding arcs and the starting position of Slow's arc.
 Their distance is given by $\sqrt{(Fast_x - Slow_x)^2 + (Fast_y - Slow_y)^2}$.
\end{paragraph}

\begin{paragraph}{Phase IIb.}
 While Slow continues on the same arc and so its coordinates remain the same as in Phase IIa, Fast is now traversing the $CB$ chord.
 Its corresponding position is
 $$\left(x_C + s \frac{t - 1 - \frac{x_1}{s}}{x_2} (x_B - x_C), y_C + s \frac{t - 1 - \frac{x_1}{s}}{x_2} (y_B - y_C)\right)$$
 where we take into account the direction from $C$ to $B$, the  starting point $C$, the speed and the initial time step.
 The normalization factor $x_2$ provides us with an actual distance instead of a percentage.
 The above results to $(Fast_x, Fast_y)$ being
 $$\left(\cos(s-1 + x_1) + s \frac{t - 1 - \frac{x_1}{s}}{x_2} (1 - \cos(s-1 + x_1)), \sin(s-1 + x_1) + s \frac{t - 1 - \frac{x_1}{s}}{x_2} (-\sin(s-1 + x_1))\right)$$
\end{paragraph}

\begin{paragraph}{Phase IIc.}
 Again, Slow is always on the same motion and its corresponding parametric equations do not need to change.
 Fast, on the other hand, commences a clockwise traversal on $\wideparen{BD}$ and so
 $$(Fast_x, Fast_y) = \left(\cos\left(2\pi -  s\left(t - 1 -\frac{x_1+x_2}{s}\right)\right), \sin\left(2\pi -  s\left(t - 1 -\frac{x_1+x_2}{s}\right)\right)\right)$$
 where Fast commences from position $2\pi$ on the boundary and moves clockwise with speed $s$ after time step $1 + \frac{x_1+x_2}{s}$.
\end{paragraph}

By studying the numerical bounds we obtain via the Fast-Chord method, we state the following result, in comparison to the other two strategies researched in this paper.

\begin{theorem}
 Fast-Chord performs better than (Generalized) Half-Chord for $s \in (c_{1.71}, c_{2.07})$.
 It also performs better than Both-to-the-Same-Point for $s \ge c_{1.71}$.
\end{theorem}

% ---------------------------------------------------------------------

%%%%%%%%%%%%%%%%%%%%%% [\THE NEW UPPER BOUND] %%%%%%%%%%%%%%%%%%%%%%%%%%%%%%%%%%%%%%%%%%%

\section{Lower Bounds}\label{sec:lower}

The main tool behind our lower bounds is the following lemma from \cite{s1}.

\begin{lemma}[Lemma 5 \cite{s1}]\label{5}     
 Consider a boundary of a disk whose subset of total length $u+\epsilon > 0$ has not been explored for some $\epsilon > 0 $ and
 $\pi \ge u > 0$. Then there exist two unexplored boundary points between which the distance along the boundary is at least $u$.
\end{lemma}

\subsection{Fast Explores}

\begin{lemma}\label{thm:fes_lb}
 Any $FES$-strategy takes at least
 \begin{itemize}
  \item $\frac{1+2\pi}{s}$ time for any $s \in [1, 2]$ and
  \item $\frac{1+2\arccos\left(-\frac{2}{s}\right)}{s} + \sqrt{1 - \frac{4}{s^2}}$ time for any $s \ge 2$.
 \end{itemize}

\end{lemma}

\begin{proof}
 
 To start with, any $FES$ strategy takes at least $\frac{1 + 2\pi}{s} - \epsilon$ time (for $\epsilon > 0$), 
 since an adversary may choose to place the exit at a point just before the time when Fast finishes exploring the whole boundary.
 This yields a lower bound of $\lim\limits_{\epsilon \rightarrow 0} \left(\frac{1 + 2\pi}{s} - \epsilon\right) = \frac{1 + 2\pi}{s}$ for any $s \ge 1$.
   
 We now show a better bound for $s \ge 2$.
 At time  $\frac{1+a}{s}$ (where $a \ge 0$), Fast has explored at most an $a$ part of the boundary.
 Then, if we consider the time  $\frac{1+a- \epsilon}{s} $ (where $\epsilon > 0$), a $2\pi - (a - \epsilon) = 2\pi - a + \epsilon$ subset of the boundary has not yet been explored. 
 We bound $a \in [\pi, 2\pi)$ such that $0 < 2\pi - a \le \pi$ holds.
 We now apply Lemma \ref{5} with $u = 2\pi - a$ and $\epsilon$.
 Thence, there exist two unexplored boundary points between which the distance along the boundary is at least $u$.
 Let us now consider the perpendicular bisector of the chord connecting these two points.
 Depending on which side of the bisector Slow lies, an adversary may place the exit on the boundary point lying at the opposite side.
 The best case for Slow is to lie exactly on the point of the bisection.
 That is, Slow will have to cover a distance of at least $\frac{2\sin\left(\frac{u}{2}\right)}{2} = \sin\left(\frac{u}{2}\right) = \sin\left(\frac{2\pi - a}{2}\right) = \sin\left(\pi - \frac{a}{2}\right) = \sin\left(\frac{a}{2}\right)$, where $2\sin\left(\frac{u}{2}\right)$ is the chord length.
 In this case, the overall evacuation time is equal to $\frac{1+a}{s} + \sin\left(\frac{a}{2}\right)$ and for the best lower bound we compute
 $ \max\limits_{\pi \le a < 2\pi} \left\{\frac{1+a}{s} + \sin\left(\frac{a}{2}\right)\right\} $.
 The rest of the proof reduces to computing the maximum of this function with respect to $a$.
 %The first partial derivative is equal to:
 %$$\frac{\partial f(s,a)}{\partial a} = \frac{1}{s} + \frac{1}{2}\cos\left(\frac{a}{2}\right)$$
 %and the family of solutions to $\frac{\partial f(s,a)}{\partial a} = 0$ is of the form:
 %$$ \left\{ 4\pi n \pm 2\arccos\left(-\frac{2}{s}\right): n \in \mathbb{Z} \right\} $$
 %The only solution which is included in the interval $[\pi, 2\pi)$ is
 %$$ a' = 2\arccos\left(-\frac{2}{s}\right) $$ 
 %and it is defined \emph{only for} $s \ge 2$.
 %Moreover, $a'$ is a local maximum, since $\frac{\partial f(s, a)}{\partial a} \mid_{a = a'} < 0$ for any $s \ge 2$.
 %It then suffices to compare $f(s, a')$ to $f(s, \pi)$ and $f(s, 2\pi)$ to prove global optimality.
 %The resulting lower bound is 
 %$$ f(s, a') = \frac{1+2\arccos\left(-\frac{2}{s}\right)}{s} + \sqrt{1 - \frac{4}{s^2}}$$
 Finally, notice that the latter bound is equal to $\frac{1 + 2\pi}{s}$ for $s = 2$ and greater than $\frac{1 + 2\pi}{s}$ for $s > 2$.
\end{proof}

\subsection{Both Explore}

\begin{lemma}\label{thm:bes}
 Any $BES$-strategy takes at least
 \begin{itemize}
  \item $1 + \frac{2}{s}\sqrt{1-\frac{s^2}{(s+1)^2}} + \frac{-s+2\arccos\left(-\frac{s}{s+1}\right) + 1}{s+1}$ time for $s \in [1, 2)$,
  \item $1 + \sqrt{1 - \frac{4}{(s+1)^2}} + \frac{-s + 2\arccos\left(-\frac{2}{s+1}\right)+1}{s+1}$ for $s \in [2, c_{4.84}]$ (where $c_{4.84} \approx 4.8406$) and
  \item $1 + \sin\left(\frac{s-1}{2}\right)$ time for $s \in (c_{4.84}, 2\pi+1)$.
 \end{itemize}
\end{lemma}     

\begin{proof}
 At time $1$, Fast has explored at most $s-1$ distance on the boundary, since it needs $\frac{1}{s}$ time to reach the boundary and in the remaining $\frac{s-1}{s}$ time it can traverse $s \frac{s-1}{s} = s-1$ distance.
 At time $1+y$, where $y \ge 0$ is a variable, Fast has explored at most an $s-1 + sy$ part of the boundary and Slow has explored at most a $y$ part of the boundary.
 We derive an upper bound for the variable $y$ by noticing that the whole explored part can be strictly less than $2\pi$ (otherwise the exit has already been found):
 $s - 1 + (s+1)y < 2\pi \Rightarrow y < \frac{2\pi - s + 1}{s+1}$.
 Then, the unexplored part is strictly greater than $2\pi - s + 1 - (s+1)y$.
 Notice that we need $s < 2\pi + 1$, otherwise we get $y < 0$ which contradicts the $y \ge 0$ initial statement.
 We let $u = 2\pi - s + 1 - (s+1)y$, where $u$ is the quantity from Lemma \ref{5}.
 We apply the restriction that $u = 2\pi - s + 1 - (s+1)y \le \pi$, which holds for $y \ge \frac{\pi - s + 1}{s+1}$.
 Moreover, $u = 2\pi - s + 1 - (s+1)y > 0$ holds for any $s \ge 1$ given that $y < \frac{2\pi - s + 1}{s+1}$.
 
 Now, let us apply Lemma \ref{5}:
 There exist two unexplored points with arc distance $\ge 2\pi - s + 1 - (s+1)y$, which implies that the chord between them
 has length at least $2\sin\left(\frac{2\pi - s - (s+1)y + 1}{2}\right) = 2\sin\left(\frac{s + (s+1)y - 1}{2}\right)$.
 An adversary can put the exit on any of the two endpoints.
 If Slow reaches an endpoint first (case I), then the exit is placed on the other side, such that Slow has to traverse the chord.
 If Fast reaches an endpoint first, then the exit is placed either on the other side (case II), meaning that Fast has to traverse the chord,
 or on the endpoint that lies the farthest from Slow current position (case III), meaning that Slow has to traverse at least half the chord.
 We assume that both the robots and the adversary behave optimally. Hence, the robots will always avoid case I.
 Then, the adversary will apply case II, for $s \in [1, 2)$, and III for $s \ge 2$.
 Let $y_{min} = \max\{0, \frac{\pi-s+1}{s+1}\}$ and $y_{max} = \frac{2\pi-s+1}{s+1}$.
 Totally, the worst-case evacuation time is given by
 \begin{itemize}
  \item $\max_{y \in \left[y_{min}, y_{max}\right)} \left\{1+y+\frac{2}{s}\sin\left(\frac{s + (s+1)y - 1}{2}\right)\right\}$, when in case II and
  \item $\max_{y \in \left[y_{min}, y_{max}\right)} \left\{1+y+\sin\left(\frac{s + (s+1)y - 1}{2}\right)\right\}$, when in case III.
 \end{itemize}
 The rest of the proof reduces to computing the maximum of these functions, with respect to $y$.
 %Initially, we compute the first partial derivative:
 %$$\frac{\partial f(s,y)}{\partial y} = 1+ (s+1)\cos((s + (s+1)y - 1)/2) $$
 % The family of solutions to $\frac{\partial f(s,y)}{\partial y} = 0$ is of the form:
 % $$ \left\{ \frac{4\pi n - s \pm 2\arccos(1/(-s-1)) + 1}{s+1}: n \in \mathbb{Z} \right\} $$
 % The only solution which is included in the interval $[y_{min}, y_{max})$ is
 %The maximum is attained onto the derivative root
 %$ y' = \frac{-s + 2\arccos(1/(-s-1)) + 1}{s+1} $
 %\emph{only for} $1 \le s \le 4.5068$
% Notice that  $\frac{\partial^2 f(s, y)}{\partial y} \mid_{y = y'} < 0$ for any $s \ge 1$.
% So far, we showed that $y'$ is the point of a local maximum.
% It suffices to show that $f(s, y') \ge f(s, y_{min})$ and $f(s, y') \ge f(s, y_{max})$ to prove global optimality.
% The resulting lower bound is $1 + \frac{2\sqrt{2s+1}}{s(s+1)} + \frac{-s+2\arccos\left(-\frac{s}{s+1}\right) + 1}{s+1}$, for $s \le 2$
% , $1 + \sqrt{1 - \frac{4}{(s+1)^2}} + \frac{-s + 2\arccos\left(-\frac{2}{s+1}\right)+1}{s+1}$ for $s \in [2, 4.8406]$ and $1 + \sin\left(\frac{s-1}{2}\right)$ time for $s \in (4.8406, 2\pi+1)$.
\end{proof}

\begin{comment}
One may observe that the above lower bound, although it is very general and strong for small values of $s$, loses its value as $s$ grows, since it diminishes to $1$.
This happens due to the fact that in the proof we consider only a specific part of a both-explore strategy, where Slow explores some part of the boundary and then it traverses a chord to reach the exit.
This chord length grows smaller and smaller for bigger values of $s$.
We understand that there is a need to capture a lower bound for the case where the slow robot has not explored any part of the boundary yet.
This is possible via \autoref{thm:fes_lb}, since we can apply the same lower bound idea also to the both-explore strategies when $s$ is big enough.

\begin{lemma}\label{thm:antipodal}
 Any $BES$-strategy takes at least $\frac{1+2\arccos\left(-\frac{2}{s}\right)}{s} + \sqrt{1 - \frac{4}{s^2}}$ time for $s \ge \pi+1$.
\end{lemma}

\begin{proof}
 One need only notice that, for $a = s-1 \ge \pi$, at time  $\frac{1 + a -\epsilon}{s}$, a $2\pi - a + \epsilon$ part of the boundary is yet unexplored, where  $2\pi - a \le \pi$.
 Moreover, the slow robot has not reached the boundary yet.
 Hence, we can see this as a fast-explores case and follow the proof of \autoref{thm:fes_lb}.
\end{proof}
\end{comment}

One may observe that the above lower bound, although it is quite strong for small values of $s$, loses its value as $s$ grows.
This happens due to the fact that in the proof we consider only a specific moment of a both-explore strategy, where both robots have already explored some part of the boundary.
Thus, there is a need to capture a lower bound for the case where Slow has not explored any part of the boundary yet.
This is possible, since we can apply a similar fast-explores lower bound idea also to the both-explore strategies when $s$ is big enough.

\begin{lemma}\label{thm:antipodal}
 Any $BES$-strategy takes at least
 \begin{itemize}
  \item $1 + \sin\left(\frac{s-1}{2}\right)$ time for $s \in (\pi+1, c_{4.97})$, where $c_{4.97} \approx 4.9699$, and
  \item $\frac{1+2\arccos\left(-\frac{2}{s}\right)}{s} + \sqrt{1 - \frac{4}{s^2}}$ time for $s \ge c_{4.97}$.
 \end{itemize}
\end{lemma}

\begin{proof}
 One need only notice that, for $a = s-1 > \pi$, at time  $\frac{1 + a -\epsilon}{s}$, a $2\pi - a + \epsilon$ part of the boundary is yet unexplored, where  $2\pi - a \le \pi$.
 Moreover, Slow has not reached the boundary yet. Hence, we can view this as a fast-explores subcase.
 Then, we can compute $\max_{a \in [\pi, \min\{s-1, 2\pi\}]} \left\{\frac{1+a}{s} + \sin\left(\frac{a}{2}\right)\right\}$.
 Due to the upper bound change for $a$, the analysis provides a $1 + \sin\left(\frac{s-1}{2}\right)$ lower bound for $s \in (\pi+1, c_{4.97}]$ and
 the already visited $\frac{1+2\arccos\left(-\frac{2}{s}\right)}{s} + \sqrt{1 - \frac{4}{s^2}}$ for $s \ge c_{4.97}$.
\end{proof}
% 
% The following lemma encompasses the above $BES$ lower bounds by taking the maximum for each value of $s$.
% 
% \begin{lemma} \label{thm:bes_lb}
%  Any $BES$-strategy takes at least
%  \begin{itemize}
%   \item $1 + \frac{2}{s}\sqrt{1-\frac{s^2}{(s+1)^2}} + \frac{-s+2\arccos\left(-\frac{s}{s+1}\right) + 1}{s+1}$ time for $s \in [1, 2)$,
%   \item $1 + \sqrt{1 - \frac{4}{(s+1)^2}} + \frac{-s + 2\arccos\left(-\frac{2}{s+1}\right)+1}{s+1}$ for $s \in [2, c_{4.84}]$,
%   \item $1 + \sin\left(\frac{s-1}{2}\right)$ time for $s \in (c_{4.84}, c_{4.97})$ and
%   \item $\frac{1+2\arccos(-2/s)}{s} + \sqrt{1-\frac{4}{s^2}}$ time for $s \in [c_{4.97}, \infty)$.
%  \end{itemize}
% \end{lemma}

\subsection{An Improvement for Both Explore}
We now obtain numerical values for a stronger $BES$ lower bound by performing a more complex analysis on the \emph{Original $BES$} lower bound proof given in Lemma \ref{thm:bes}.
The main idea behind the improvement is to provide a better bound for case III of the proof, i.e. when the adversary places the exit on the farthest endpoint from Slow's current position.
Apparently, the best play for Slow is to lie exactly on the midpoint of the chord with the unexplored endpoints.
Nevertheless, in order for Slow to be there, it needs to spend some of its time, originally destined for exploration, within the disk interior.
We hereby examine the best possible scenario for Slow in terms of its distance from the midpoint following the above reasoning.
Let us refer to this lower bound as \emph{Improved} $BES$.

\begin{lemma}\label{thm:besNEW}
 Improved $BES$ is greater or equal to Original $BES$ for any $s \ge 1$.
\end{lemma}     

\begin{proof}
 %At time $1$, Fast has explored at most $s-1$ distance on the boundary, since it needs $\frac{1}{s}$ time to reach the boundary and in the remaining $\frac{s-1}{s}$ time it can traverse $s \frac{s-1}{s} = s-1$ distance.
 At time $1+y$, where $y \ge 0$ is a variable, Fast has explored at most an $s-1 + sy$ part of the boundary and Slow has explored at most a $y$ part of the boundary.
 Now, in extension to the previous lower bound, suppose that Slow has spent $k$ time, where $k \in [0, y]$, \emph{not exploring} the boundary, i.e. moving within the disk interior.
 % not that he reaches boundary at least one, otherwise bla bla bla
 
 Notice that it takes $1 + \frac{2\pi - s + 1}{s+1}$ time for the whole perimeter to be explored, when both robots are exploring after timestep $1$ (a subcase of $BSP$ bounds).
 Thence, we upper-bound $y \le \frac{2\pi - s + 1}{s+1}$.
 To lower-bound $y$, we restrict the unexplored part $u = 2\pi - s + 1 - (s+1)y + k \le \pi$.
 That is, we get $y \ge \max\{\frac{\pi-s+1+k}{s+1}, 0\}$.
 Moreover, $u > 0$ is already covered by the aforementioned upper bound.
 
 Now, we are ready to apply Lemma \ref{5}:
 There exist two unexplored points (say $A, B$) with arc distance $\ge 2\pi - s + 1 - (s+1)y + k$, which implies that the chord between them
 has length at least $2\sin\left(\frac{2\pi - s + 1 - (s+1)y + k}{2}\right) = 2\sin\left(\frac{s - 1 + (s+1)y - k}{2}\right)$.
 An adversary could place the exit on any of the two endpoints.
 If Slow reaches an endpoint first (case I), then the exit is placed on the other side, such that Slow has to traverse the chord.
 If Fast reaches an endpoint first, then the exit is placed either on the other side (case II), meaning that Fast has to traverse the chord,
 or on the endpoint that lies the farthest from Slow's current position (case III), meaning that Slow has to traverse at least half the chord.
 We assume that both the robots and the adversary behave optimally. Hence, the robots will always avoid case I.
 
 Let us now examine more carefully what happens in case III.
 For a depiction of the proof, see \autoref{fig:newLB}.
 The ideal location for Slow is to lie exactly on the chord midpoint, say $M$.
 Nevertheless, this may not be possible due to it only spending $k$ time within the disk interior.
 Let us consider the minimum distance from the chord midpoint to the boundary.
 This is exactly $1 - \lambda$, where $\lambda = |OM|$ is the distance from the midpoint to the center of the disk.
 Notice that $OM$ intesects $AB$ \emph{perpendicularly}, since $M$ is the midpoint of chord $AB$.
 Using the Pythagorean theorem in $\triangle AMO$, we get $\lambda = \sqrt{1 - \sin^2\left(\frac{s - 1 + (s+1)y - k}{2}\right)} = \left|\cos\left(\frac{s - 1 + (s+1)y - k}{2}\right)\right|$.
 If we consider the case when $1 - \lambda > k$, then the ideal position for Slow is to lie $k$ distance away from the boundary and on the extension of $OM$ (i.e. on point $K$).
 From there, Slow can take a beeline to the exit yielding a $\sqrt{\sin^2\left(\frac{s - 1 + (s+1)y - k}{2}\right) + (1 - \lambda - k)^2}$ distance again by the Pythagorean theorem, now in $\triangle AMK$.
 
 To conclude, Slow will try to minimize this beeline distance over $k$, while the adversary will select a case between II and III that maximizes the total distance.
 Overall, the optimization problem reduces to computing:
  
  \begin{equation}
  \max\limits_{y \in \left[y_{min}, y_{max}\right)} \left\{ 1 + y + \max
  \left\{
  \begin{array}{l}
   \min\limits_{k \in [0, y]}\frac{2}{s}\sin\left(\frac{s -1+ (s+1)y - k}{2}\right), \\
   \min\limits_{k \in [0, y]} \sqrt{\sin^2\left(\frac{s - 1 + (s+1)y - k}{2}\right) + \max\left\{1 - \lambda - k, 0\right\}^2} \\
  \end{array}
  \right\} \right\}
  \end{equation}
 
 Note that the above bound matches the original one for the cases where $1 - \lambda < k$.
 
 Last but not least, we need also consider the case where the adversary chooses to place the exit on the last boundary point to be explored.
 In the current setting, it takes at least $\frac{u}{s+1} = \frac{2\pi - s + 1 - (s+1)y + k}{s+1}$ extra time for both robots to explore the rest of the boundary,
 since Fast explores $s\frac{u}{s+1}$ while Slow explores $\frac{u}{s+1}$ for a total distance of $u$.
 Overall, we are looking to compute:
 
 \begin{equation}
  \max\limits_{y \in \left[y_{min}, y_{max}\right)} \min\limits_{k \in [0, y]} \left\{1+y+\frac{2\pi - s + 1 - (s+1)y + k}{s+1} \right\} = 
  \max\limits_{y \in \left[y_{min}, y_{max}\right)}  \left\{1+y+\frac{2\pi - s + 1 - (s+1)y}{s+1} \right\}
 \end{equation}

 Due to the inherent complexity of the optimization problem in (1), we compute \emph{numerical} bounds.
 The expressions (1) and (2) are computed and the maximum of them is chosen as the best-play scenario for an adversary.
 The computational work is done in Matlab. 
 We iterate over feasible values of variables $y$ and $k$ with a step of $10^{-3}$.
 For Fast's speed $s$, we iterate with a step of $10^{-1}$.
 The resulting experimental bounds show that, for all $s \in [1, 2\pi+1)$, 
 this lower bound is greater or equal to the lower bound given in Lemma \ref{thm:bes}.
 \end{proof}

 The numerical results we get from the above proof demonstrate the two robots \emph{always} choose $k = 0$ as the value of their minimizer.
 That is, it appears that spending any of their ``exploration time'' off the boundary should not provide any assistance to the robots.
 
 \begin{figure}
  \begin{center}
 \includegraphics{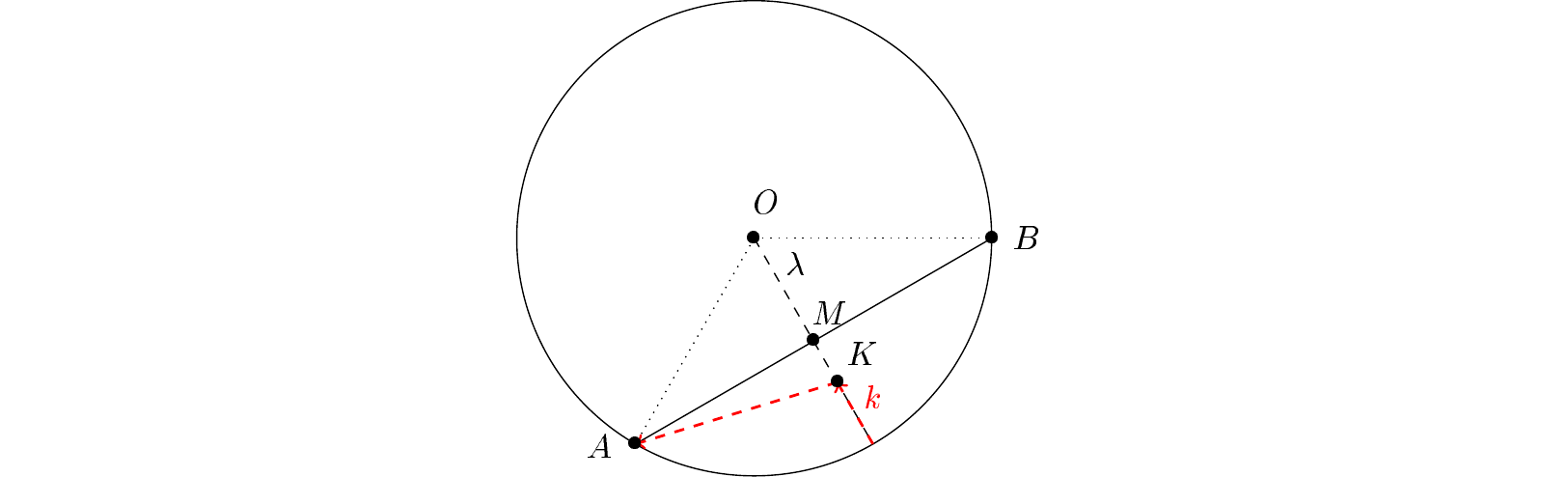}
 \end{center}
 \caption{An Improved Lower Bound for $BES$}
 \label{fig:newLB}
 \end{figure}

\section{Comparison}
\begin{comment}
We notice that for values of $s$ close to $1$, the $BES$ lower bound is much lower to the $FES$ one. 
On the other hand, the $FES$ lower bound is a more reasonable estimate and prevails for $s \ge 1.856$ showing that for that speed range a ``fast-explores'' approach would be a better option. 
\end{comment} 
 
\subsection{Lower Bounds} 

For each value of $s$ we select the minimum (weakest) lower bound between the (maximum) $BES$ and $FES$ ones as our overall lower bound; see \autoref{fig:lbCOMP}.
We see that Improved $BES$ is stronger than Original $BES$ one for any $s \ge c_{1.71} \approx 1.71$.
Moreover, Improved $BES$ is stronger than the $FES$ lower bound for $s \ge c_{2.75} \approx 2.75$.

\begin{figure}[H]
 \begin{center}
 \includegraphics{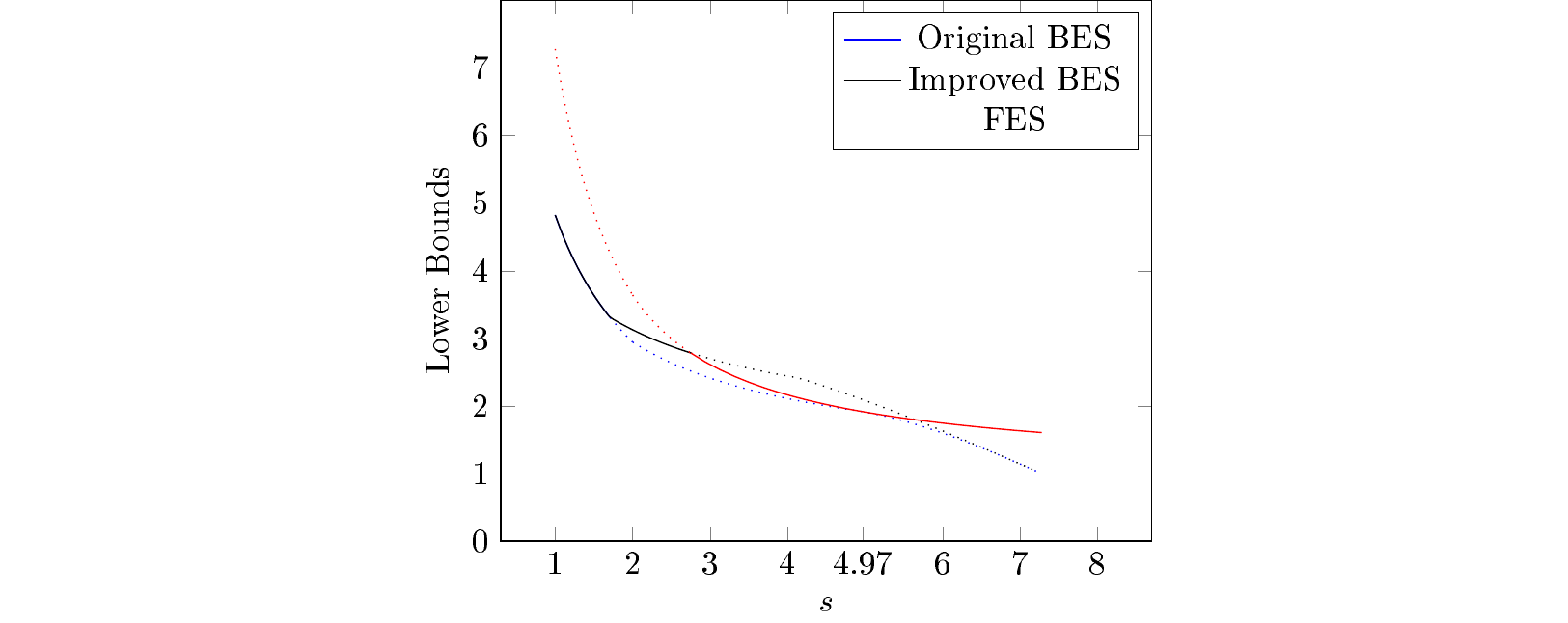}
 \end{center}
 \caption{Comparison of lower bounds}
 \label{fig:lbCOMP}
\end{figure}

% \begin{theorem} \label{thm:lb}
%  Any evacuation strategy takes at least 
%  \begin{itemize}
%   \item $1 + \frac{2}{s}\sqrt{1-\frac{s^2}{(s+1)^2}} + \frac{-s+2\arccos\left(-\frac{s}{s+1}\right) + 1}{s+1}$ time for $s \in [1, 2)$,
%   \item $1 + \sqrt{1 - \frac{4}{(s+1)^2}} + \frac{-s + 2\arccos\left(-\frac{2}{s+1}\right)+1}{s+1}$ for $s \in [2, c_{4.84}]$,
%   \item $1 + \sin\left(\frac{s-1}{2}\right)$ time for $s \in (c_{4.84}, c_{4.97})$ and
%   \item $\frac{1+2\arccos(-2/s)}{s} + \sqrt{1-\frac{4}{s^2}}$ time for $s \in [c_{4.97}, \infty)$.
%  \end{itemize}
% \end{theorem} 

\subsection{Upper Bounds}
We notice that Half-Chord outperforms BSP for any $s \ge c_{1.86} \approx 1.856$.
Besides, Fast-Chord outperforms BSP for any $s \ge c_{1.71} \approx 1.71$.
Finally, Fast-Chord outperforms Half-Chord for any $s \le c_{2.07} \approx 2.072$.
That is, the introduction of Fast-Chord yields a better upper bound for any $s \in [c_{1.71}, c_{2.07}]$.

\begin{figure}[H]
 \begin{center}
 \includegraphics{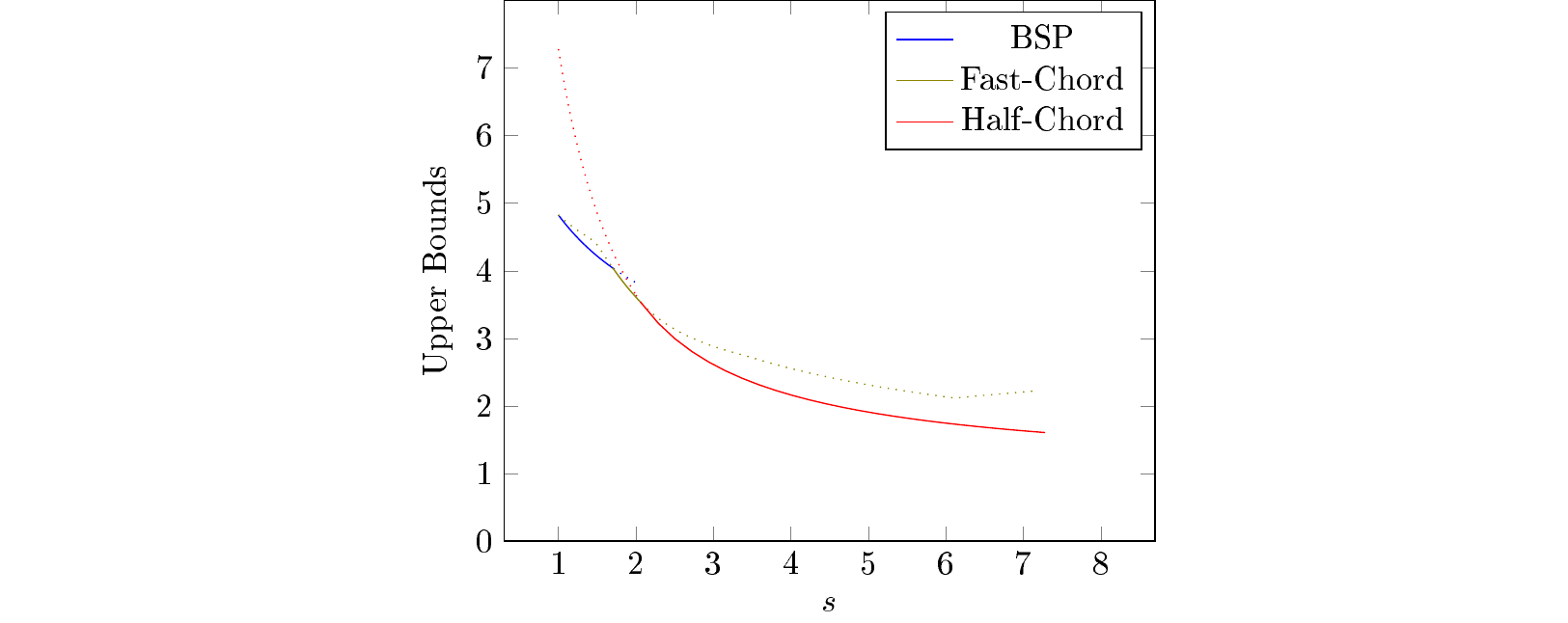}
 \end{center}
 \caption{Comparison of upper bounds}
 \label{fig:upper}
\end{figure}

\subsection{Comparison of Bounds}
 
By comparing upper and lower bounds, we see that Half-Chord is optimal for $s \geq c_{2.75}$, since the matching $FES$ lower bound is the weakest in this interval.
On the other hand, for $s < c_{2.75}$ the ratio between the bounds is at most $1.22$ (maximized when $s = c_{1.71}$), where
the strategy changes from BSP to Fast-Chord.
The best strategy to use is $BSP$ when $s < c_{1.71}$, Fast-Chord when $c_{1.71} < s < c_{2.07}$ and Half-Chord for $s \ge c_{2.07}$. 

\begin{figure}[H]
 \begin{center}
 \includegraphics{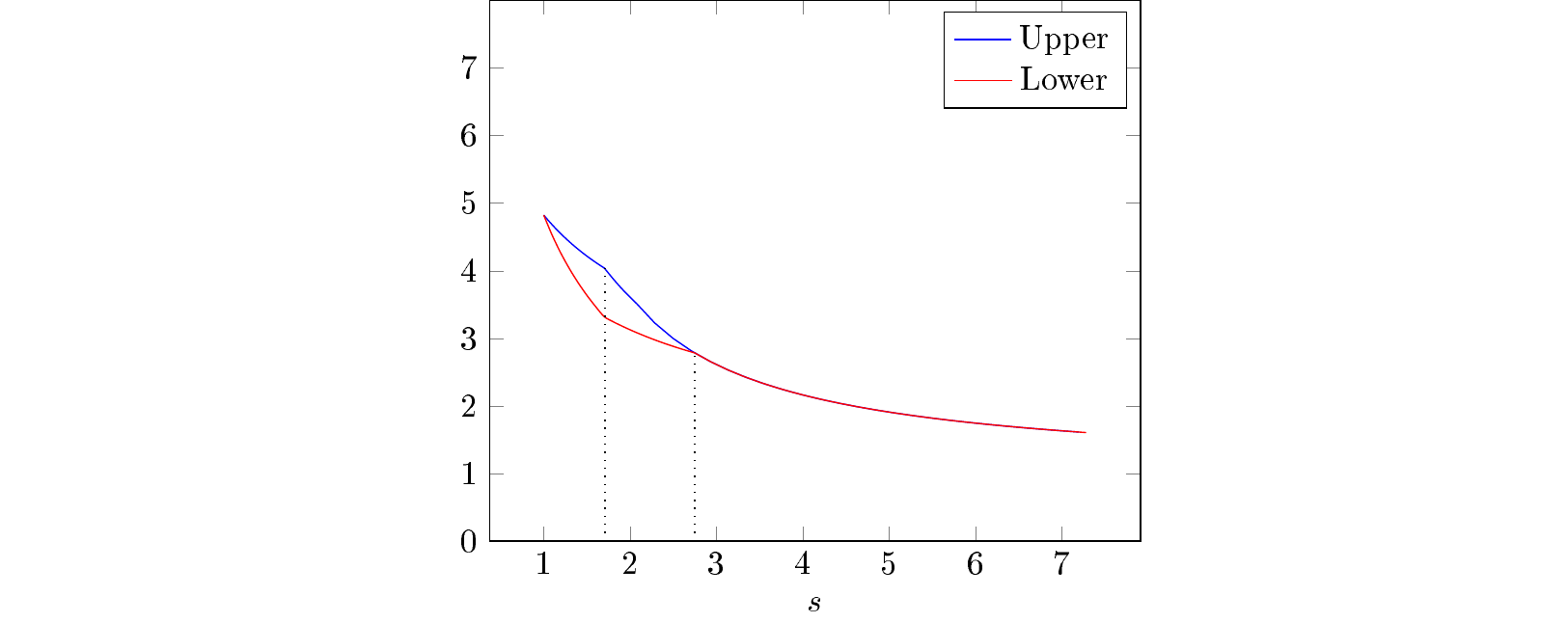}
 \end{center}
 \caption{Dominant Lower vs Upper Bounds}
 \label{fig:COMP}
\end{figure}

\section{Open Questions}
Optimality for the case $1 < s < c_{2.75}$ remains open.
Regarding further work on this topic, one could consider extending these results to a more-than-two-robots evacuation scenario.
Moreover, the non-wireless case for two-robots fast evacuation seems to be a quite challenging open problem 
given the fact that exact optimality appears to be complex to obtain even for $s = 1$ (\cite{face-to-face}).


\newpage 

\begin{thebibliography}{99} 

\small

\bibitem{alpern02b}
S.~Alpern and S.~Gal.
{\em The Theory of Search Games and Rendezvous.}
Int.\ Series in Operations Research and Management Science,
number 55, Kluwer Academic Publishers, 2002.
 
\bibitem{BCR88} 
R.A.~Baeza-Yates, J.C.~Culberson, and G.J.E.~Rawlins.
Searching with uncertainty. 
{\em Proc.\ SWAT'88: 1st Scandinavian workshop on algorithm theory},
{\bf 318}, pp.~176--189.  

\bibitem{BCR93}
R.A.~Baeza-Yates, J.C.~Culberson, and G.J.E.~Rawlins.
Searching in the plane.
{\em Inform.\ and Comput.} {\bf 106} (1993), pp.\ 234--252.

\bibitem{BS95}
R.A.~Baeza-Yates and R.~Schott.
Parallel searching in the plane.  
{\em Computational Geometric Theory and Applications},
{\bf 5} (1995), pp.~143--154.

\bibitem{B64}
A.~Beck. 
On the linear search problem. 
{\em Naval Res.\ Logist.} {\bf 2} (1964), pp.\ 221--228.

\bibitem{B56}
R.~Bellman.
Minimization problem.  
{\em Bull.\ AMS} {\bf 62} (1956), p.\ 270.
 
\bibitem{B63}
R.~Bellman. 
An optimal search problem. 
{\em SIAM Rev.} {\bf 5} (1963), p.\ 274.

% \bibitem{Bender}
% M.A.~Bender, A.~Fern\'{a}ndez, D.~Ron, A.~Sahai, and S.P.~Vadhan.
% The power of a pebble: Exploring and mapping directed graphs.  
% {\em STOC 1998}, pp.~269--278.

\bibitem{Bonato}
A.~Bonato and R.J.~Nowakowski.
{\em The Game of Cops and Robbers on Graphs.}
American Mathematical Society, 2011.

\bibitem{BDD}
P.~Bose, J.-L.~De Carufel, and S.~Durocher.
Revisting the problem of searching a line.  
{\em ESA 2013}, pp.~205--216.

\bibitem{YC}
Y.~Chevaleyre.
Theoretical analysis of the multi-agent patrolling problem.
{\em IAT 2004}, pp.\ 302--308.

\bibitem{CGGM}
M.~Chrobak, L.~G\k{a}sieniec, T.~Gorry, and R.~Martin.
Group search on the line.  
{\em SOFSEM 2015}, pp.~164--176. 

\bibitem{s1}
J.~Czyzowicz, L.~G\k{a}sieniec, T.~Gorry, E.~Kranakis, R.~Martin, and D.~Paj\k{a}k, 
Evacuating robots from an unknown exit located on the boundary of a disc,
\emph{DISC 2014}, LNCS 8784, pp. 122-136, Springer, 2014.

\bibitem{czyz2011}
J.~Czyzowicz, L.~G\k{a}sieniec, A.~Kosowski, and E.~Kranakis.
Bondary patrolling by mobile agents with distinct maximal speeds.
{\em ESA 2011}, pp.~701--712.

\bibitem{face-to-face}
J.~Czyzowicz, K.~Georgiou, E.~Kranakis, L.~Narayanan, J.~Opatrny, and B.~Vogtenhuber,  
Evacuating robots from a disk using face-to-face communication, \emph{CIAC 2015}, 
pp.\ 140--152.

\bibitem{Gluss}
B.~Gluss.  
An alternative solution to the ``lost at sea'' problem.  
{\em Naval Research Logistics Quarterly}, {\bf 8} (1961), pp.~117--122.

\bibitem{JL09} 
A.~Je\.{z}, and J.~\L{}opuza\'{n}ski.
On the two-dimensional cow search problem.
{\em Information Processing Letters} {\bf 131} (2009), pp.~543--547.

\bibitem{KRT}
M.-Y.~Kao, J.H.~Reif, and S.R.~Tate.
Searching in an unknown environment: An optimal randomized 
algorithm for the cow-path problem.
{\em Inform.\ and Comput.} {\bf 131} (1996), pp.\ 63--80.  

\bibitem{KDM}
E.~Kranakis, D.~Krizanc, and E.~Markou.
{\em The Mobile Agent Rendezvous Problem in the Ring.} 
Morgan \& Claypool Publishers, 2010.  

\bibitem{KKR}
E.~Kranakis, D.~Krizanc, and S.~Rajsbaum.
Mobile agent rendezvous: A survey.
{\em SIROCCO 2006}, pp.\ 1--9.    

\bibitem{LC09}
H.~Li and K P.~Chong.
Search on lines and graphs.
{\em Proc.\ 48th IEEE Conference on Decision and Control, 
2009 held jointly with the 2009 28th 
Chinese Control Conference. (CDC/CCC 2009)} {\bf 109(11)}, 
pp.~5780--5785.

\bibitem{Parsons}
T.D.~Parsons.
Pursuit-evasion in a graph, in 
{\em Theory and Applications of Graphs: Proc., Michigan May 11--15}, pp.\ 426--441, 1976. 

\bibitem{TF10} T.~Temple and E.~Frazzoli.
Whittle-indexability of the cow path problem.
{\em American Control Conference (ACC) 2010}, 
pp.~4152--4158.

\bibitem{YWB}
V.~Yanovski, I.A~Wagner, and A.M.~Bruckstein.
A distributed ant algorithm for efficiently patrolling a network.
{\em Algorithmica} {\bf 37} (2003), pp.\ 165--186. 

\end{thebibliography}
\end{document}